%% file: USNKPCPaper.tex
\documentclass[a4paper]{article}

\usepackage{graphicx}
\usepackage{amsmath, amssymb}
\usepackage{amsthm}
\usepackage{enumitem}
\usepackage{parskip}
\usepackage{setspace}
\usepackage{thmtools}
\declaretheoremstyle[
headfont=\bfseries,%
headpunct={.}, %
numbered=yes,
spaceabove=10pt, %
postheadspace=10pt ]{bettertheorem}
\declaretheorem[name=Theorem,style=bettertheorem]{theorem}

\declaretheorem[name=Assumption,style=bettertheorem]{assumption}
\declaretheorem[name=Example,style=bettertheorem]{example}

\usepackage{mathtools}

\DeclarePairedDelimiter\floor{\lfloor}{\rfloor}
\newcommand{\overbar}[1]{\mkern 1.5mu\overline{\mkern-1.5mu#1\mkern-1.5mu}\mkern 1.5mu}

\usepackage{mathrsfs}
\usepackage{listings}

\DeclareOldFontCommand{\bf}{\normalfont\bfseries}{\mathbf}
\usepackage{geometry}
\usepackage{appendix}
\usepackage{verbatim}
\usepackage{adjustbox}
\usepackage{float}
\usepackage[format = hang]{subcaption}
\usepackage[format = hang]{caption}
\usepackage{siunitx}
\usepackage{listings}
\usepackage[hang,flushmargin]{footmisc}
\usepackage{lscape}
\usepackage{wrapfig}	
\usepackage{multirow}

\usepackage[style=authoryear-comp, citetracker = true, uniquename=false, natbib, maxbibnames=99, uniquelist = false]{biblatex}
\bibliography{USNKPC}
\usepackage[nottoc,numbib]{tocbibind}
\AtEveryBibitem{\clearfield{month}}

\usepackage{pgfplots}
\usepackage{booktabs}
\usepackage{array}
\usepackage{graphicx}
\graphicspath{{TablesFigures/}}
\usepackage[ruled]{algorithm2e}
\SetKw{Break}{break}
 
 \newcolumntype{P}[1]{>{\centering\arraybackslash}p{#1}}
 \newcolumntype{L}[1]{>{\raggedright\arraybackslash}p{#1}}
\newcolumntype{M}[1]{>{\centering\arraybackslash}m{#1}}

\title{Inference on the New Keynesian Phillips Curve with Very Many Instrumental Variables}

\author{
  Max-Sebastian Dovì\thanks{This research is funded by the German National Merit Foundation and the European Research Council via Consolidator grant number 647152. I thank Sophocles Mavroeidis and Anna Mikusheva for very helpful comments and suggestions. I also thank seminar participants at the University of Oxford, and participants at the 2019 European Conference of the Econometrics Community. All errors and omissions are my own.}\\
  max-sebastian.dovi@economics.ox.ac.uk\\
}

\begin{document}
\maketitle

\begin{abstract}

	\noindent Limited-information inference on New Keynesian Phillips Curves (NKPCs) and other single-equation macroeconomic relations is characterised by weak and high-dimensional instrumental variables (IVs). Beyond the efficiency concerns previously raised in the literature, I show by simulation that ad-hoc selection procedures can lead to substantial biases in post-selection inference. I propose a Sup Score test that remains valid under dependent data, arbitrarily weak identification, and a number of IVs that increases exponentially with the sample size. Conducting inference on a standard NKPC with 359 IVs and 179 observations, I find substantially wider confidence sets than those commonly found.\\
		
	\end{abstract}

\pagebreak

\include{USNKPCBody}

\pagebreak

\printbibliography

\appendix\newpage\markboth{Appendix}{Appendix}
\renewcommand{\thesection}{\Alph{section}}
\numberwithin{equation}{section}

\include{USNKPCAppendix}

\end{document}

%% file: USNKPCBody.tex
\section{Introduction}

Instrumental variable (IV) methods are often used to conduct limited-information inference on (structural) single-equation macroeconomic relations that describe the dependence of a scalar variable on a set of covariates. Examples of such macroeconomic relations include New Keynesian Phillips Curves (NKPCs), Euler equations, and Taylor rules. IV-based limited-information inference on such macroeconomic relations has arguably proven popular because there is no requirement that parts of the model other than the specified relation itself be necessarily true to conduct valid inference. In virtually all applications, the relation is assumed to contain an additive error term that is shown (e.g., by the assumption of Rational Expectations (RE)) or primitively assumed to be uncorrelated with predetermined variables excluded from the specified relation. This makes any predetermined variable a valid IV.

As documented extensively in the existing literature, using IVs to conduct limited-information inference on such macroeconomic relations often runs into issues related to weak identification. This occurs when the variation in the IVs is only able to explain a small portion of the variation of the endogenous variables.\footnote{For the case of NKPCs, see \citet{Kapetanios:2015fp, Mirza:2014wd, Mavroeidis:2014ge, Kleibergen:2009do, Dufour:2006bh, Ma:2002wl}, for the case of Euler equations see \citet{Ascari:2019ch, Kleibergen:2005ug, Yogo:2004vm, Stock:2000vt}, for the case of Taylor rules see \citet{Mirza:2014wd, Mavroeidis:2010bv}.} This problem is especially pronounced when the analysis is restricted to using only a few variables to forecast the endogenous variables, a restriction that arises when using IV methods that treat the number of IVs as fixed relative to the sample size. Since any predetermined variable is a valid (if not very informative) IV, this naturally raises the question of which IVs to choose out of the very many available ones. 
 
The limited literature that seeks to formally address the high dimensionality of the available IVs in such macroeconomic settings is primarily motivated by the potential inefficiency of using IVs selected in an ad-hoc way \citep{Berriel:2019vp, Bayar:2018er, Berriel:2016hj, Mirza:2014wd, Kapetanios:2015fp}. Through simulations and/or empirical applications, these studies find smaller confidence sets than the ones implied by IVs traditionally used in the past.  Although this evidence is certainly suggestive, it should be noted that formal efficiency claims rely on conditions that are not easily verifiable in practice.\footnote{For instance, factor-based approaches to reduce the dimensionality of the IVs likely work well if there is a factor structure, and if whatever explains most of the variation in the IVs, also explains (a good portion of) the variation of the endogenous variables. While the former may be made plausible through certain tests, the latter remains an assumption the researcher has to make. Similarly, a LASSO-based selection of IVs works well only under the assumption that the relation between the endogenous variables and the candidate IVs is sufficiently sparse.}
 
Rather than being motivated by such efficiency concerns, this paper revisits the question of high-dimensional limited-information inference because some types of formal or intuitive regularisation can lead to invalid inference, even if weak-IV robust methods are used after regularisation. This is due to what \citet{Chernozhukov:2015iq} call the `endogeneity bias', which arises when variables are selected on the basis of their in-sample correlation with a model's error terms. 

The first contribution of this paper consists in illustrating how improper selection of IVs can lead to invalid inference in the context of limited-information inference on a standard NKPC. I do this by extending the simulations in \citet{Mavroeidis:2014ge} to the more realistic case where the econometrician does not have oracle knowledge on which IVs are the relevant ones, but rather has to choose amongst the very many available IVs. I consider different IV selection techniques, and show that several of them result in substantially invalid inference. The example of the NKPC is chosen for the sake of concreteness and due to its popularity in the literature. The same concerns extend to any of the many cases in Macroeconomics where a given structural equation can be estimated with very many valid IVs.

As a second contribution, I propose a Sup Score test to conduct IV-based limited-information inference on single-equation Macroeconomic relations. Contrarily to other approaches in the literature, this statistic requires no assumption on the factor structure of the IVs, nor does it make any sparsity-type assumption that requires only a few of the very many IVs to be relevant, while allowing for a number of IVs that increases exponentially with the sample size. This test directly contributes to the (very) many weak IVs literature predominantly restricted to the cross-sectional case (see \citet{Mikusheva:2020uo, Crudu:2020bk, Belloni:2012kw, Anatolyev:2010kk}), and can find application well beyond the example of NKPCs considered in this paper. 

The third contribution consists in applying the selection procedures considered in the simulation section and the Sup Score test to conduct IV-based limited-information inference on a standard hybrid NKPC with 359 IVs on a sample of 179 observations. I find that both the IVs selected and the confidence set implied by the selection procedure that yields the worst size distortion in the simulations are similar to the IVs selected and the confidence set implied by the IVs traditionally used in the past. This suggests that the results previously reported in the literature may suffer from endogeneity bias, and that they hence may undercover the true parameter values. By contrast, the confidence sets implied by the Sup Score test are considerably wider. 

\textbf{Notation}. For any real number $a$, $\floor*{a}$ indicates the smallest integer $b$ such that $b \leq a$. For any two real numbers $c$ and $d$, $c \lesssim d$ if $c$ is smaller than or equal to $d$ up to a universal positive constant. The remaining notation follows standard conventions.

\textbf{Organisation of the paper}. Section \ref{Section-Model} introduces the model considered in this paper. Section \ref{Section-Methodology} outlines the methods used in this paper to conduct inference in the context of very many IVs. Section \ref{Section-Simulation} provides simulation-based evidence on the size and power of these methods. Section \ref{Section-Empirics} revisits inference on the US NKPC using very many IVs. Section \ref{Section-Conclusion} concludes. 

\section{Model \label{Section-Model}}

The structural equation I consider is the hybrid NKPC of \citet{Gali:1999tx},
\begin{equation}
	\label{Equation-PC-Hybrid} 
	 	\pi_t = c + \lambda s_t + \gamma_f\mathbb{E}_t\left[\pi_{t+1}\right] +\gamma_b\pi_{t-1} + u_t,
\end{equation}
where $\pi_t$ is the inflation rate, $s_t$ is the forcing variable, and $\lambda$, $c$, $\gamma_{f}$, and $\gamma_b$ are parameters of the model. $u_t$ is an unobserved disturbance term, which can be interpreted as a measurement error, or as a shock to inflation, such as a cost-push shock.

The identifying moment conditions can be derived within the framework of Generalised Instrumental Variable (GIV) estimation. In this approach, realised one-period-ahead inflation is substituted in for expected inflation. This means that Equation \eqref{Equation-PC-Hybrid} can be re-written as
\begin{equation*}
	\label{Equation-PC-Hybrid-GIV}
				\pi_t = c + \lambda s_t + \gamma_{f}\pi_{t+1} + \gamma_{b}\pi_{t-1}+ \underbrace{u_t - \gamma_{f}\left[\pi_{t+1} - \mathbb{E}_t[\pi_{t+1}]\right]}_{{\epsilon}_t}.
\end{equation*}

If it is further assumed that $\mathbb{E}_{t-1}[u_t] = 0$, the assumption of RE gives rise to the moment conditions
\begin{equation*}
	\label{Equation-MC-Basic}
	\mathbb{E}[{Z}_t{\epsilon}_t] = 0,
\end{equation*}
for any $k\times 1$ vector of predetermined variables ${Z}_t$. Due to the very large number of predetermined time series available, the dimension of $Z_t$ is comparable to or larger than the number of observations, $T$.

It should be noted that the example of NKPCs (including the particular specification chosen), and the assumption of RE are not central to two of the contributions of this paper. The same concerns relating to the endogeneity bias persist, and the same Sup Score test proposed below remains valid for the broad class of models defined by single-equation relations of the type
	\begin{equation}
		\label{Equation-General-Struct}
		 y = g( Y,  X, \theta) + \varepsilon,
	\end{equation}
	and moment equations given by	
	\begin{equation}
		\label{Equation-General-MC}
		\mathbb{E}[{Z}_t\varepsilon_t] = 0,
	\end{equation}	
where $ y$ is a $T\times 1$ vector, $g$ is a known real-valued function, $ Y$ is a $T\times p_1$ matrix of endogenous covariates, $ X$ is a $T\times p_2$ matrix of exogenous covariates, $Z$ is a $T\times k$ matrix of variables such that $k \geq p_1 + p_2$, $\theta$ is a $(p_1 + p_2)\times 1$ vector of coefficients, $p_1$ and $p_2$ are both fixed, and $ \varepsilon$ is a $T\times 1$ vector of error terms. This setup encompasses many popular applications in Macroeconomics, where $k$ is of the same magnitude or even larger than $T$, such as limited-information inference on NKPCs, Euler equations, and Taylor rules.

In particular, the NKPC considered in Equation \eqref{Equation-PC-Hybrid-GIV} can be mapped into the more general model in Equation \eqref{Equation-General-Struct} as follows. Since the NKPC is linear, the exogenous (predetermined) variables can be partialled out. Hence, $y = M_X\pi$,  $Y = M_X[s \text{ } \pi_{+1}]$, $M_X = I - X(X'X)^{-1}X'$, $ X = [{1}_{T\times 1} \text{ } \pi_{-1}]$, $g( Y,  X, \theta) = Y\theta$, $\theta = [\lambda, \gamma_f]'$, $\varepsilon = M_X\epsilon$,  $Z = M_X\tilde{Z}$, $\tilde{Z}$ is a $T\times (k-2)$ matrix of excluded IVs, $s, \pi_{+1}$, $\pi_{-1}$, and $\epsilon$ are the $T\times 1$ stacked vectors of $s_t$, $\pi_{t+1}$, $\pi_{t-1}$, and $\epsilon_t$, respectively. 

\section{Methodology \label{Section-Methodology}}

For all methods considered in this paper, confidence sets are constructed by inverting statistics that test the hypothesis
\begin{equation}
	\label{Equation-Null-Hypothesis}
			H_0: \theta = \theta_0 \text{  vs   } H_1: \theta \neq \theta_0.
\end{equation}

The $(1-\alpha)$ confidence set can be constructed by collecting the values of $\theta_0$ for which the null hypothesis in Equation \eqref{Equation-Null-Hypothesis} is not rejected at the $\alpha$ level of significance. For convenience, define ${\varepsilon}_0 \equiv  y - g( Y,  X, \theta_0)$.

\subsection{Post-Selection Low-Dimensional Inference \label{Section-Low-Dim}}

Most of the existing literature that conducts inference on relations of the form presented in Equation \eqref{Equation-General-Struct} using moment conditions of the type shown in Equation \eqref{Equation-General-MC} has employed methods that require the IVs to be low-dimensional. In the presence of very many IVs, these approaches can be seen as a two-step procedure. First, the IVs are selected. Second, a low-dimensional (weak-identification robust) method is applied with the selected IVs. The first step is usually not made explicit, and is often not given any attention, which makes it impossible to model this step accurately. In Section \ref{Section-Methodology-Selection}, I consider three different selection procedures that reasonably cover (in terms of their deleterious effect on subsequent inference) the range of selection procedures used in the previous literature. These are random selection, `crude thresholding', and LASSO. In Section \ref{Section-Methodology-S}, I outline the $S$ statistic of \citet{Stock:2000vt}, which forms the post-selection inferential method common to all three selection procedures considered in this paper.

 Before proceeding, it is helpful to gain some intuition as to why IV selection may lead to invalid IVs. For simplicity, suppose that all variables are endogenous (or that the model is linear and that the exogenous covariates have been partialled out). Consider the following projection (`first stage')
\begin{equation*}
	Y = Z\zeta + v,
\end{equation*}
where $\zeta$ is a $k\times 1$ vector of coefficients and $v$ is a $T\times 1$ vector of error terms. Consider the case of no identification at all, $\zeta = 0$, and a selection procedure that selects the IVs that are most highly correlated with the endogenous variables, $Y$. This amounts to selecting those IVs that are most highly correlated in-sample with the first-stage error term. By the endogeneity of the system, this means that those IVs most correlated with the error term, $\varepsilon$, will be selected, so that \emph{conditional on selection}, the IVs are no longer valid. This phenomenon carries over more broadly to cases of weak (but non-zero) identification as discussed in \citet{Hansen:2014ie}.

\subsubsection{The \citet{Stock:2000vt} $S$ Statistic \label{Section-Methodology-S}}

In this paper, the GMM-based $S$ statistic of \citet{Stock:2000vt} will be used for low-dimensional post-selection inference.\footnote{More powerful and computationally intensive (GMM-based) weak-identification robust methods could be used instead of the $S$ statistic (see \citet{Mirza:2014wd, Kleibergen:2009do}). Considering them instead of the $S$ statistic does not qualitatively affect the results of the simulations, while increasing their computational burden substantively. Furthermore, \citet[p. 165]{Mavroeidis:2014ge} state that amongst the different specifications for the NKPC they consider, the confidence sets implied by these more powerful methods are similar to the ones implied by the $S$ statistic.} Letting $k_s \geq p_1 + p_2$ denote the number of IVs selected, the $S$ statistic is given by $T$ times the value of the continuously updated GMM objective function given by
\begin{equation}
	\label{Equation-S-Stat}
	S(\theta_0) = T\varepsilon_T(\theta_0)'W_T(\theta_0)\varepsilon_T(\theta_0),
\end{equation}
where $\varepsilon_T(\theta_0) = T^{-1}\sum_{t = 1}^TZ_{st}\varepsilon_{0t}$, $Z_{st}$ is the $k_s\times 1$ vector containing the IVs selected, and $W_T(\theta_0)$ is the continuously updated $k_s\times k_s$ weight matrix that is a consistent estimator of the covariance matrix of the moment conditions of the selected IVs as in \citet{Kleibergen:2009do, Stock:2000vt}. Throughout, I use the heteroscedasticity and autocorrelation consistent (HAC) estimator of \citet{Newey:1987ua}. Under the null hypothesis in Equation \eqref{Equation-Null-Hypothesis} and the regularity conditions discussed in \citet{Stock:2000vt}, this statistic is asymptotically $\chi^2_{k_s}$. Whenever $p_2 \neq 0$ (i.e., there are exogenous covariates in the relation), the exogenous covariates can be concentrated out, to yield the concentrated $S$ statistic as in \citet[Theorem 3]{Stock:2000vt}.\footnote{In both the simulations and the empirical application below, the constant and the one-period lagged inflation are concentrated out.}

The $S$ statistic further recommends itself in this context because it allows for a straightforward test of the exclusion restrictions of the IVs. It may be hoped that any substantial bias caused by improper selection may be flagged in the form of a low $p$-value for the test of the null hypothesis that the IVs selected, $Z_s$, are uncorrelated with the structural error term, $\varepsilon$. To investigate this possibility further, in the simulations, I also evaluate the weak-identification robust Hansen test. This is given by the minimum value of the $S$ statistic in Equation \eqref{Equation-S-Stat}. Without making an assumption of strong identification, this statistic is asymptotically bounded by a $\chi^2_{k_s-p_2}$ distribution \citep[p. 178]{Mavroeidis:2014ge}, which provides a weak-identification robust critical value for the test of the overidentifying restrictions of the IVs selected.

\subsubsection{Ad-Hoc Selection of Instrumental Variables \label{Section-Methodology-Selection}}

Conducting inference with the $S$ statistic requires selecting a sufficiently small subset of $k_s$ IVs from the available $k$ IVs.\footnote{An often-used rule of thumb is to select $k_s$ to be of the order of magnitude of $T^{1/3}$. This rate result is motivated by the results in \citet{Andrews:2007bl}, and \citet{Newey:2009fs}, who show that this rate condition is sufficient for the case of independent data. Recently, fully weak-identification robust AR-type statistics have been developed that allow for the number of IVs to be of the order of magnitude of $T$ \citep{Mikusheva:2020uo, Crudu:2020bk, Anatolyev:2010kk}. However, all of these approaches treat the IVs as fixed, and are hence not applicable in the context of time series. \label{Footnote-What-Is-Big}} In most of the empirical studies on IV-based limited-information inference on macroeconomic relations, no explicit reason is given for choosing the $k_s$ IVs that are subsequently used for analysis. Often, the choice of IVs is simply motivated with reference to previous studies that used those IVs. It is hence impossible to model the choice of IVs of the previous literature accurately in a simulation exercise. As an (imperfect) approximation, I consider the following three selection procedures.

The first selection procedure involves randomly selecting $k_s$ IVs out of the $k$ available IVs. Since the selection of IVs is not informed by the data itself, this selection procedure is guaranteed to not violate the identifying moment conditions. 

The second selection procedure I consider will be referred to as crude thresholding. This involves first computing $p_1$ separate $k\times 1$ vectors containing the sample correlations between the endogenous variables and all the candidate IVs, sorting the IVs in descending order of correlation, and constructing the vector of IVs for post-selection inference by taking the union of the first $\floor*{k_s/p_1}$ entries in each of the vectors. By selecting the variables based on in-sample correlations, this selection procedure is likely to break the exclusion restriction of the IVs selected. Although (to my knowledge) this crude thresholding has not been applied to IV-based limited-information inference, more sophisticated versions of thresholding have been considered in the past (e.g., \citet[Appendix B]{Mirza:2014wd} and \citet{Bayar:2018er}).\footnote{The hard thresholding in \citet[Appendix B]{Mirza:2014wd} and \citet{Bayar:2018er} is not applicable in high-dimensional contexts, since OLS is infeasible when there are more variables than observations.}

The first two selection procedures (random selection and crude thresholding) arguably cover the extremes in terms of the effects IV selection can have on the validity of the IVs. Random selection provides the selection ideal, since it leaves the identifying moment conditions completely unaffected. However, particularly with reference to the traditional IVs often considered in the literature, it seems unlikely that random selection (over the very many available predetermined macroeconomic time series) led to choosing proximate lags of the endogenous variables as IVs. Indeed, given the persistence of most macroeconomic time series (and hence of the endogenous variables in any given application), it seems plausible that at least part of the motivation for considering proximate lags of the endogenous variables as IVs stems from their ability to usefully explain some of their in-sample variation. Suggestive evidence for this type of selection is also given by the fact that the IVs selected by crude thresholding in the empirical application in Section \ref{Section-Empirics} show substantial overlap with these traditional IVs.\footnote{See also the ranking of IVs based on $t$-values in \citet[Table B1]{Mirza:2014wd}.} Therefore, it seems likely that random selection and crude thresholding provide a suggestive lower and upper bound on the selection-induced bias that could underlie existing empirical applications.

The third selection procedure I consider is a LASSO-based selection of IVs. This is motivated by the recent increase in popularity of penalisation-based approaches to the (very) many IV problem (see \citet{Hansen:2014ie, Belloni:2012kw, Ng:2011eq}). Furthermore, LASSO-based approaches to IV selection have also been applied to the case of NKPCs in \citet{Berriel:2019vp, Berriel:2016hj}. Here, IVs are selected by solving a LASSO optimisation problem of the following form for each endogenous variable
\begin{equation*}
	\label{Equation-LASSO-Naive}
	\begin{aligned}
	\hat{\zeta}_r&=\underset{\zeta_r \in \mathbb R^k}{\text{ arg min } } \sum_{t=1}^T (Y_{rt} - \zeta_r'Z_t)^2 + \Lambda_r|\zeta_r|,\\
	\end{aligned}
\end{equation*}
where $Y_{rt}$ is the element in position $t$ of the $T\times 1$ vector $Y_r$ given by the $r^{th}$ column of $Y$, $\zeta_r$ for $r = 1, \dots, p_1$ is a $k\times 1$ vector, and $\Lambda_r > 0$ for $r = 1, \dots, p_1$ are scalar penalty parameters that are set such that $\hat{\zeta}_r$ has $\floor*{k_s/p_1}$ elements. The IVs selected are given by the IVs that have at least one corresponding non-zero entry in at least one of $\hat{\zeta}_r$ for $r = 1, \dots, p_1$. 

\subsection{A High-Dimensional Sup Score Test for Dependent Data}

Although the interplay between weak identification and high-dimensional IVs has recently received some attention (see \citet{Hansen:2014ie}), none of the currently available approaches are both robust to arbitrarily weak identifcation and applicable in a time-series context. Indeed, to the best of my knowledge, the only approach that is formally robust to arbitrarily weak identification in the presence of very many IVs is the Sup Score test of \citet{Belloni:2012kw}. The Sup Score test of \citet{Belloni:2012kw}, however, treats the IVs as fixed, and is hence not applicable in time-series contexts. In this section, I propose a Sup Score test that remains valid under high-dimensional dependent data using recent results of \citet{Zhang:2018hy, Zhang:2014vm}.

The Sup Score statistic I propose is given by
\begin{equation}
	\label{Equation-TSS}
	\mathcal{R} = \underset{1 \leq j \leq k}{\text{ max }} \left| \frac{1}{\sqrt T} {Z}_{j}'{\varepsilon}_0\right|.
\end{equation}

This can be seen as a non-studentised version of the \citet{Belloni:2012kw} Sup Score statistic, which in turn can be interpreted as an extension to high dimensions of the \citet{Anderson:1949tx} (AR) statistic. It also bears some resemblance to the non-studentised AR statistic proposed by \citet{Horowitz:2018to}.

The critical values for the test statistic in Equation \eqref{Equation-TSS} are computed using a block bootstrap. Let $l_T \equiv \floor*{T / b_T}$, where $b_T$ is the block length. Define the block sums
\begin{equation*}
	\hat{A}_{tj} = \sum_{l = (t-1)b_T + 1}^{tb_T}{Z}_{lj}{\varepsilon}_{0l} - \left\{\overbar{Z'\varepsilon_0}\right\}_j,	 \text{ for } t = 1, \dots, l_T,
\end{equation*}
where $\left\{\overbar{Z'\varepsilon_0}\right\}_j$ is the $j^{th}$ element of the $k\times 1$ vector $\frac{1}{T}\sum_{t = 1}^TZ_t\varepsilon_t$. Consider the bootstrap statistic given by
\begin{equation*}
	L_{\hat A} = \underset{1 \leq j \leq k}{\text{ max }} \frac{1}{\sqrt T}\left|\sum_{t = 1}^{l_T}\hat{A}_{tj}e_t\right|,
\end{equation*}
 where $\{e_t\}$ is a sequence of i.i.d. $\mathcal{N}[0, 1]$ random  variables. The critical value for a test of size $\alpha$ of Equation \eqref{Equation-Null-Hypothesis} is given by
 \begin{equation*}
 	c(\alpha) = \text{inf}\left\{\gamma \in \mathbb{R}:\mathbb{P}(L_{\hat A} \leq \gamma| \{{Z}_t{\varepsilon}_{0t}\}_{t = 1}^T) \geq 1-\alpha\right\}.
 \end{equation*}
 
The decision rule for testing the null hypothesis in Equation \eqref{Equation-Null-Hypothesis} at the $\alpha$ level of significance is given by
 \begin{equation*}
 	\text{Reject } H_0 \iff \mathcal{R} > c(\alpha).
 \end{equation*}
 
 I now turn to conditions that are sufficient to ensure that the test described above has correct size. 
 
 \begin{assumption} $\left.\right.$
 	\label{Assumption-Main}
 	 \begin{enumerate}[label = \roman*.]
 	 	\item $Z_t\varepsilon_t$ is a stationary time series that allows for the causal representation $Z_{t}\varepsilon_t = \mathcal{G}(\dots, u_{t-1}, u_t)$ for some measurable function $\mathcal G$, where $u_t$ are a sequence of mean-zero i.i.d. random variables. Furthermore, assume that $Z_{tj}\varepsilon_t = \mathcal{G}_j(\dots, u_{t-1}, u_t)$ for all $j = 1, \dots, k$, where $\mathcal{G}_j$ is the $j$th component of the map $\mathcal G$. \label{Assumption-Z-Rep}
 	 \item $\mathbb{E}[Z_t\varepsilon_t] = 0$, $\mathbb{E}[Z_{tj}^2\varepsilon_t^2] > 0$, and $\mathbb{E}[Z_{tj}^4\varepsilon_t^4] < \infty$ for all $j = 1, \dots, k$. \label{Assumption-Z}
 	\item $k \lesssim \text{exp}(T^b)$, $b_T  \lesssim T^{\tilde{b}}$ for $b < 1/15$, $4\tilde{b} + 7b < 1$, $\tilde{b} - 2b > 0$. \label{Assumption-Dim}
	\item $\underset{1 \leq j, h \leq k}{\text{ max }}\sum_{l = -\infty}^\infty |l| \mathbb{E}[|Z_{tj}\varepsilon_{t}Z_{t+l, h}\varepsilon_{t+l}|] = O(T^{\breve{b}})$, $\breve{b} < \tilde{b} - 2b$. \label{Assumption-Weird-MC}
	\item $\mathbb{E}[|\mathcal{G}_j(\dots, u_{t-1}, u_t) - \mathcal{G}_j(\dots, u^*_{-1}, u^*_0, u_1, \dots, u_t)|^q] \leq C \rho^t$, for some $0 <\rho < 1$, and some positive constant $C$, where $q \geq 4$, and $\{u^*_t\}$ are i.i.d. copies of $\{u_t\}$. \label{Assumption-UGMC}
\end{enumerate}

\end{assumption}

Assumption \ref{Assumption-Main}.\ref{Assumption-Z-Rep} requires the product of the IVs and the error terms to be stationary, and have some causal representation. Assumption \ref{Assumption-Main}.\ref{Assumption-Z} makes weak assumptions on the moments of the data, and includes the identifying moment condition. In practice, I standardise the IVs in-sample to ensure that the test is invariant to the scaling of IVs. Assumption \ref{Assumption-Main}.\ref{Assumption-Dim} bounds the degree of high dimensionality permitted and the size of the block bootstraps. Although the restriction on the dimensionality ($b < 1/15$) is stronger than the ones usually encountered in the independent case (see \citet{Deng:2020wk, Belloni:2012kw}), it still allows for very many IVs compared to the sample size. Assumption \ref{Assumption-Main}.\ref{Assumption-Weird-MC} imposes restrictions on the correlation of the product of the IVs with the error term across different points in time. Assumption \ref{Assumption-Main}.\ref{Assumption-UGMC} imposes a (uniform) Geometric Moment Contraction (GMC) restriction on the product of the IVs and the error terms as in \citet{Wang:2019ta}. The GMC requires that the process under consideration have a sufficiently `short memory'. Processes that obey such a condition include (under suitable assumptions) standard linear processes (e.g., standard vector autoregressions and Volterra processes) as well as several nonlinear processes (e.g., autoregressive models with conditional heteroscedasticity, random coefficient autoregressive models, and exponential autoregressive models).  I refer to \citet{Wang:2019ta, Zhang:2018hy, Chen:2016fh, Zhang:2014vm, Wu:2005uh, Hsing:2004gma} and the references therein for a discussion of the different processes that obey such a condition.

No assumption on the first stage (i.e., the relationship between $Y$ and $Z$) has to be made. This means that the proposed Sup Score test is uniformly valid over all (finite) values of the coefficient on the IVs in the first stage (including arbitrarily weak identification). This also means that no restriction on the factor or sparsity structure of the first stage has to be imposed. The lack of assumptions on the first stage also implies that the Sup Score test does not suffer from any `missing IV problem' (see also \citet{Dufour:2009bu}).

Whether these conditions are satisfied in any given macroeconomic application depends on the error terms (i.e., the structural equation), and on the properties of the excluded IVs. Example \ref{Example-Z-UGMC} shows that under suitable assumptions on the error term that encompass, amongst others, some popular assumptions made in the literature on NKPCs (e.g., \citet{Dufour:2006bh}), it is only required that the IVs satisfy a GMC condition. This is attractive in the context of limited-information inference, since the researcher only has to assume that the IVs belong to one of the many processes that have been shown to obey such a condition, without having to take a stance on the particular process. 

\begin{example}
	\label{Example-Z-UGMC}
	
	Assume that $\varepsilon_t$ is i.i.d. across $t$, $\mathbb{E}[\varepsilon_t] = 0$,  $\mathbb{E}[\varepsilon_t^2] > 0$, and $\mathbb{E}[\varepsilon_t^4] < \infty$. Assume that $Z_t$ is a stationary time series and allows for the causal representation $Z_t = \mathcal{F}(\dots, v_{t-1}, v_t), Z_{tj} = \mathcal{F}_j(\dots, v_{t-1}, v_t)$ for some measurable function $\mathcal F$, where $v_t$ are a sequence of mean-zero i.i.d. random variables (independent of $\varepsilon_t$). Assume further that $\mathbb{E}[Z_{tj}^2] >0$ and $\mathbb{E}[Z_{tj}^4] < \infty$ for all $j = 1, \dots, k$. Assume that the conditions on the dimensionality of the IV problem in Assumption \ref{Assumption-Main}.\ref{Assumption-Dim} are satisfied. Further, assume that $Z_t$ satisfies:	
	\begin{equation}
		\label{Equation-Z-UGMC}
		\mathbb{E}[|Z_{tj} - \mathcal{F}_j(\dots, v^*_{-1}, v_0^*, v_1, \dots, v_t)|^4] < \tilde{C}{\tilde{\rho}}^{t}
	\end{equation}
	where $\{v^*_t\}$ are i.i.d. copies of $\{v_t\}$, $\tilde{C}$ is some constant, and $0 <\tilde{\rho}
	 < 1)$. Then the conditions in Assumption \ref{Assumption-Main}. hold. 
	
	\begin{proof}
	See Appendix \ref{Appendix-Example}.	
	\end{proof}
\end{example}

I now state the main theoretical result of this paper, which ensures that the approach proposed controls the size of the test.\footnote{It should be noted, however, that--similarly to other sup-based test statistics, such as in \citet{Chernozhukov:2018ema, Belloni:2012kw}--the above approach is not efficient. This is to be expected, given the weak assumptions made on (the structure of) the IVs. The (finite-sample) power properties of the above approach will be investigated in the simulation section below. The results show that it has non-trivial power.} 

\begin{theorem}
\label{Theorem-Main}
	Under Assumption \ref{Assumption-Main}. and the null hypothesis in Equation \eqref{Equation-Null-Hypothesis},
	\begin{equation*}
		\underset{T\to \infty}{\text{ lim }}\mathbb{P}(\text{Reject } H_0) \leq \alpha. 
	\end{equation*}
\end{theorem}

\begin{proof}
	
	See Appendix \ref{Appendix-Proof}.

\end{proof}

 Theorem \ref{Theorem-Main} makes it possible to construct confidence sets by inverting the test as outlined above.

\section{Simulations \label{Section-Simulation}}

The simulations presented in this section serve a twofold purpose. First, I use the simulations to study how improper selection of IVs can lead to problematic post-selection inference. Second, I use the simulations to illustrate the asymptotic validity of the Sup Score test established in the section above, as well as its finite-sample power properties. Taken together, the simulations hence motivate and further justify applying the Sup Score test proposed in Section \ref{Section-Model} in practice.\footnote{Due to the focus of this paper on the bias introduced by IV selection, I do not consider the factor-based approaches of \citet{Kapetanios:2015fp, Mirza:2014wd}. The substantial biases caused by the improper selection of a small number of IVs can also serve to motivate the use of such factor methods. However, the factor-based GMM approach in \citet{Mirza:2014wd} seems to treat the number of IVs as fixed (and does not provide formal conditions for validity), while the factor AR statistic of \citet{Kapetanios:2015fp} is only applicable in a high-dimensional context if a sufficiently strong factor structure is assumed. In contrast, the Sup Score test proposed in this paper remains valid in high-dimensional contexts regardless of the factor or sparsity structure of the IVs.}

The simulations in this paper are based on the approach in \citet{Mavroeidis:2014ge}. The central difference is that rather than modelling the econometrician as having perfect knowledge of the relevant IVs, and incorrectly employing methods that are not robust to weak identification, I model the econometrician as using exclusively weak-identification robust methods, but not knowing which IVs correspond to the truly relevant ones. Given the extensive literature that pointed out that NKPCs can suffer from weak identification, this setup seems closer to the estimation problem that an econometrician is likely to face.

I base my simulations on the simplest possible specification considered in \citet{Mavroeidis:2014ge}. This involves imposing the restriction $\gamma_b + \gamma_f = 1$ (which is known to the econometrician) and setting $c = 0$ (which is not known to the econometrician), so that the NKPC can be re-written as
\begin{equation}
	\label{Equation-NKPC-Sim}
	(1-\gamma_f)\Delta\pi_t = \lambda s_t + \gamma_f\mathbb{E}[\Delta\pi_{t+1}]  + \epsilon_t.	
\end{equation}

I embed this NKPC into a dynamic system by specifying that the reduced-form dynamics of the forcing variable and inflation follow a VAR model given by
\begin{equation}
	\label{Equation-RF-VAR}
	\begin{aligned}
	\begin{bmatrix}
		\pi_t\\
		s_t\\
		f_t
	\end{bmatrix} &= \begin{bmatrix}
		a_{11} & a_{12} & a_{13}\\
		a_{21} & a_{22} & a_{23}\\
		a_{31} & a_{32} & a_{33}\\
	\end{bmatrix}\begin{bmatrix}
		\pi_{t-1}\\
		s_{t-1}\\
		f_{t-1}
	\end{bmatrix} + \begin{bmatrix}
		u_{1t}\\
		u_{2t}\\
		u_{3t}
	\end{bmatrix},
	\end{aligned}
\end{equation}
where
\begin{equation*}
	\begin{bmatrix}
		u_{1t}\\
		u_{2t}\\
		u_{3t}
	\end{bmatrix}\overset{i.i.d.}{\sim} \mathcal{N}\left[0, \begin{bmatrix}
 	\omega_{11} & \omega_{12} & \omega_{13}\\
 	\omega_{21} & \omega_{22} & \omega_{23}\\
 	\omega_{31} & \omega_{32} & \omega_{33}\\
 \end{bmatrix}\right],
\end{equation*}
and $f_t$ is a scalar factor variable. All coefficients except for $a_{11}, a_{12}$, and $a_{13}$ have to be calibrated. The coefficients $a_{11}, a_{12}$, and $a_{13}$ are backed out of the NKPC based on the \citet{Anderson:1985wr} algorithm.

High dimensionality of the IVs is introduced by specifying that there exists an $m\times 1$ vector of variables $Q_t$ that follow the process given by
\begin{equation}
	\label{Equation-Q}
	Q_t = \xi f_t + u_{4t}, u_{4t} \overset{i.i.d.}{\sim} \mathcal{N}\left[0, I_{m}\right],
\end{equation}
where $\xi$ is an $m\times 1$ vector of factor loadings.

The econometrician conducts inference on $\lambda$ and $\gamma_f$ within the GIV and RE framework,
\begin{equation}
	\label{Equation-Estimate-Econometrician}
	\begin{aligned}
	\Delta \pi_t &= c + \lambda s_t + \gamma_f(\pi_{t+1}-\pi_{t-1}) + \epsilon_t,\\
	\mathbb{E}[Z_{st}\epsilon_t] &= 0,
	\end{aligned}
\end{equation}
where the variables are defined as in Section \ref{Section-Model} and Section \ref{Section-Methodology}. The econometrician does not observe the factor itself, but only observes the forcing variable and inflation, as well as the $m$ variables in $Q_t$. In this setup, $Z_{st}$ is a subset of the available IVs given by $Z_t = [1, \pi_{t-1}, s_{t-1}, Q_{t-1}']'$ that always includes a constant (since it is specified in the structural equation the econometrician considers).\footnote{For the case of post-selection inference based on the $S$ statistic, the constant is concentrated out. For the case of the Sup Score statistic, it is partialled out.}

This setup recommends itself for two reasons. First, it constitutes a minimal departure from popular simulations in the existing literature. This ensures that any reported results are not an artefact of a particularly uncharitable setup.\footnote{It is, for instance, straightforward to include the variables in $Q_t$ directly in the reduced-form VARs. However, the results from such a DGP are very sensitive to the particular calibration of the parameters chosen.} Second, it ensures the existence of a sufficiently small set of (excluded) `oracle IVs' (given by $\pi_{t-1}, s_{t-1}, f_{t-1}$) without necessarily imposing a sparse setup on the observed first-stage projection (although it can be imposed by setting $a_{13} = a_{23} = \omega_{13} = \omega_{31} = \omega_{23} = \omega_{32} = 0$ or simply $\xi = 0$).\footnote{Ensuring a sufficiently sparse set of `oracle IVs' further motivates considering only a single lag of a single factor.} The former is desirable because it allows for a comparison of the different ad-hoc inference procedures relative to the most efficient approach (conditional on using the $S$ statistic). The latter is desirable because it allows for a more general (and perhaps more realistic, see \citet{Giannone:2018uv}) approach to modelling the first stage. This setup is able to achieve both a sparse (unobserved) oracle first stage and an observed first stage that is not necessarily sparse because the elements of $Q_{t-1}$ that have a non-zero corresponding entry in $\xi$ will contain some relevant variation for identification, due to the dependence of the endogenous variables $[\pi_{t+1} - \pi_{t-1}, s_t]'$ on $f_{t-1}$. Based on the setup above, it is possible to derive two different concentration parameters ($\mu^2_O$ and $\mu^2_E$) that reflect the strength of identification in the sparse unobserved oracle first stage and the observed first stage. The details are given in Appendix \ref{Appendix-Conc-Par}.

The calibrations are as follows. Throughout, I set $\gamma_f = 0.8$, $\lambda = 0.05$, $T = 100$, and $\omega_{11} = 0.07$, $\omega_{12} = \omega_{21} = 0.03$, $\omega_{22} = 0.7$, $\omega_{13} = \omega_{31} = \omega_{23} = \omega_{32} = 0$, and $\omega_{33} = 0.4$ (see also \citet{Mavroeidis:2014ge}). The results are not sensitive to this choice of covariance matrix, and this setup makes it possible to create a perfectly sparse observed first stage by setting $a_{23} = 0$. For simplicity, I set $a_{31} = a_{32} = 0$, so that the factor structure follows an autoregressive process with coefficient given by $a_{33} = 0.7$ (the results do no change appreciably if this is relaxed or a different choice for $a_{33}$ is considered). I set $m = 200$ and $\xi_{q} =  \tau(-1)^q\log\left((q+1)^2/mq\right)$ for $q = 1, \dots, m$ and $\tau = 0.05$. This is meant to provide a deterministic calibration that balances positive and negative, as well as large and small coefficients. The small value chosen for $\tau$ ensures that the information on the factor contained in the observed variables is sufficiently diluted, and that there is some interesting variation in the informational content of the unobserved oracle first stage and the one actually observed.\footnote{I refer to Appendix \ref{Appendix-Conc-Par} for more details on this. The derivations also show that choosing small values of $\tau$ has a similar effect to choosing a larger term for the variance of the errors in Equation \eqref{Equation-Q}.} The results are unaffected by different choices of $\xi$ or $\tau$. For all selection procedures, I force the selection of $k_s = 4$ IVs to ensure that the first stage is not overfitted, which again ensures that any distortions in inference are attributable to the selection step itself. For the $S$ statistic, I set the lag-length for the \citet{Newey:1987ua} HAC variance estimator to 4. For the Sup Score test proposed above, I set the block size to $b_T = 4$ and the bootstrap replications to $500$. I allow $a_{21}$, $a_{22}$, and $a_{23}$ to take on different values. The coefficient $a_{23}$ controls how informative the factor is in predicting the endogenous variables, and by extension how informative the variables in $Q_{t-1}$ are.

Table \ref{Table-Sim-Results} shows the size of the $S$ and Sup Score statistic following the different selection procedures outlined above for a test with nominal size 10\%. The calibrations chosen ensure that a broad range of identification strength and sparsity structures are considered. The first panel for $a_{23}$ corresponds to the perfectly sparse first stage where none of the variables in $Q_{t-1}$ are informative IVs. As a consequence, the concentration parameter of the unobserved oracle first stage is the same as the one that is observed. The second and third panel increase the dependence of the two endogenous variables on the unobserved factor. Due to the dense calibration of $\xi$, this means that all of the IVs observed by the econometrician are at least somewhat informative. Since the variables in $Q_t$ contain noisy information on the unobserved factor, the concentration parameter in the observed first stage will now be lower than the concentration parameter of the unobserved oracle first stage. As expected, the oracle IVs yield correct, if somewhat conservative, size. Since random selection does not make use of any correlations present in the actual data, the $S$ statistic with randomly selected IVs also yields correct size. The results for crude thresholding and the LASSO suggest that in all cases size is not controlled, although the distortions appear to be somewhat milder for the LASSO. The results for the Sup Score test proposed in this paper show that the test controls for size regardless of the DGP considered.

Table \ref{Table-Sim-Results} also reports the rejection frequency of a two-step approach that tests the null hypothesis at a given level of significance only if the robust test of overidentifying restrictions fails to reject the hypothesis of exogeneity for the IVs selected at that level of significance. This is a very conservative approach. Indeed, when faced with evidence that the selected IVs may be endogenous, rather than abandoning the analysis altogether, it seems more likely that the econometrician will proceed to select other IVs, potentially worsening the endogeneity bias. Even in this conservative approach, crude thresholding fails to control for size. LASSO selection followed by this two-step approach appears to control for size. These results suggest that while the test of overidentifying restrictions can help mitigate some of the endogeneity bias introduced by improper selection, it is unable to fully remove it.

Figure \ref{Figure-Sim-Results-Power} shows the power of the different approaches. I present the results for the case where $a_{21} = a_{22} = a_{23} = 0.450$. The results are similar for other calibrations. The map traced out by the oracle IVs corresponds to the most powerful procedure possible (conditional on exclusively using the $S$ statistic) that controls for size. The results show that randomly selecting IVs yields no power. This is unsurprising, given that in this setup the first stage is sparse, so that random selection predominantly selects not very informative IVs. The power heatmaps for crude thresholding and the LASSO have a similar shape to the oracle heatmaps. However, for certain parts of the parameter space considered, the rejection frequency of these procedures is substantially higher than the one of the oracle test. Conditional on using the same test post-selection, both crude thresholding and the LASSO can be at most as powerful as the test that directly uses the oracle IVs. Therefore, this excess rejection frequency is spurious, which in practice would translate to small confidence sets. The power heatmaps for the Sup Score test show that the Sup Score test has non-trivial power.

\begin{table}[H]
\scriptsize
\caption{Simulation results: size.}\label{Table-Sim-Results}
\centering
\begin{tabular}{ccccccccccc}
\toprule

& & \multicolumn{9}{c}{$a_{23} = 0.000$}\\
\cmidrule{2-11}

&  &\multicolumn{3}{c}{ $a_{21} = 0.000$ }  & \multicolumn{3}{c}{ $a_{21} =0.200$ }  & \multicolumn{3}{c}{ $a_{21} =0.450$ }\\

& $a_{22}$& $ 0.000$ &$0.200$ & $0.450$ &  $ 0.000$ &$0.200$ & $0.450$ & $ 0.000$ &$0.200$ & $0.450$\\

\cmidrule{2-11}

& $\mu^2_O$ &0.000&4.082&24.175&0.000&4.070&23.938&0.000&4.040&23.393\\
& $\mu^2_E$ &0.000&4.082&24.175&0.000&4.070&23.938&0.000&4.040&23.393\\

\cmidrule{2-11}

\multirow{2}*{Oracle} & R.F. &0.070&0.042&0.046&0.062&0.064&0.045&0.069&0.059&0.046\\
& T.S. &0.058&0.037&0.038&0.057&0.055&0.041&0.060&0.055&0.039\\
\multirow{2}*{Random} & R.F. &0.107&0.125&0.123&0.114&0.115&0.115&0.112&0.120&0.135\\
& T.S. &0.100&0.118&0.112&0.108&0.107&0.111&0.100&0.112&0.124\\
{Crude} & R.F. &0.445&0.421&0.428&0.417&0.433&0.429&0.396&0.419&0.440\\
{Thresholding} & T.S. &0.200&0.179&0.160&0.186&0.191&0.177&0.183&0.223&0.150\\
\multirow{2}*{LASSO} & R.F. &0.199&0.208&0.216&0.204&0.198&0.220&0.189&0.219&0.230\\
& T.S. &0.129&0.130&0.095&0.124&0.114&0.094&0.130&0.147&0.088\\
\multirow{2}*{Sup Score} & R.F. &0.030&0.027&0.037&0.033&0.038&0.035&0.035&0.018&0.033\\

& T.S. &$-$ & $-$ & $-$ & $-$ & $-$ & $-$ & $-$ & $-$ & $-$\\

\toprule

& & \multicolumn{9}{c}{$a_{23} = 0.200$}\\

\cmidrule{2-11}

&  &\multicolumn{3}{c}{ $a_{21} = 0.000$ }  & \multicolumn{3}{c}{ $a_{21} =0.200$ }  & \multicolumn{3}{c}{ $a_{21} =0.450$ }\\

& $a_{22}$& $ 0.000$ &$0.200$ & $0.450$ &  $ 0.000$ &$0.200$ & $0.450$ & $ 0.000$ &$0.200$ & $0.450$\\

\cmidrule{2-11}

& $\mu^2_O$ &11.198&19.501&48.211&12.137&21.030&49.983&13.431&23.283&53.527\\
& $\mu^2_E$ &6.638&14.262&38.910&7.144&15.044&38.117&7.827&16.092&36.214\\
\cmidrule{2-11}

\multirow{2}*{Oracle} & R.F. &0.059&0.051&0.058&0.057&0.053&0.042&0.063&0.035&0.056\\
& T.S. &0.041&0.041&0.047&0.042&0.042&0.033&0.043&0.026&0.045\\
\multirow{2}*{Random} & R.F. &0.108&0.106&0.107&0.121&0.114&0.100&0.125&0.112&0.129\\
& T.S. &0.102&0.103&0.096&0.113&0.107&0.093&0.113&0.109&0.116\\
{Crude} & R.F. &0.413&0.411&0.406&0.409&0.411&0.400&0.416&0.365&0.405\\
{Thresholding} & T.S. &0.176&0.172&0.171&0.198&0.192&0.193&0.208&0.197&0.219\\
\multirow{2}*{LASSO} & R.F. &0.188&0.194&0.223&0.186&0.196&0.193&0.192&0.179&0.195\\
& T.S. &0.112&0.114&0.108&0.113&0.127&0.090&0.130&0.118&0.119\\
\multirow{2}*{Sup Score} & R.F. &0.024&0.032&0.032&0.028&0.033&0.032&0.028&0.033&0.026\\
& T.S. &$-$ & $-$ & $-$ & $-$ & $-$ & $-$ & $-$ & $-$ & $-$\\

\toprule

& & \multicolumn{9}{c}{$a_{23} = 0.450$}\\

\cmidrule{2-11}

&  &\multicolumn{3}{c}{ $a_{21} = 0.000$ }  & \multicolumn{3}{c}{ $a_{21} =0.200$ }  & \multicolumn{3}{c}{ $a_{21} =0.450$ }\\

& $a_{22}$& $ 0.000$ &$0.200$ & $0.450$ &  $ 0.000$ &$0.200$ & $0.450$ & $ 0.000$ &$0.200$ & $0.450$\\  

\cmidrule{2-11}

& $\mu^2_O$ &56.388&81.202&101.025&61.410&83.525&101.183&66.590&83.989&107.289\\
& $\mu^2_E$ &29.100&42.505&52.008&30.144&41.690&47.471&31.006&40.080&42.374\\

\cmidrule{2-11}

\multirow{2}*{Oracle} & R.F. &0.084&0.067&0.056&0.067&0.063&0.081&0.056&0.065&0.087\\
& T.S. &0.060&0.051&0.042&0.049&0.048&0.065&0.037&0.052&0.072\\
\multirow{2}*{Random} & R.F. &0.107&0.131&0.101&0.115&0.105&0.123&0.131&0.125&0.122\\
& T.S. &0.097&0.121&0.095&0.108&0.102&0.112&0.119&0.117&0.105\\
{Crude} & R.F. &0.372&0.416&0.391&0.383&0.359&0.395&0.381&0.395&0.398\\
{Thresholding} & T.S. &0.186&0.203&0.180&0.182&0.180&0.210&0.191&0.204&0.292\\
\multirow{2}*{LASSO} & R.F. &0.201&0.215&0.189&0.201&0.204&0.195&0.186&0.222&0.212\\
& T.S. &0.127&0.122&0.104&0.118&0.125&0.119&0.117&0.138&0.189\\
\multirow{2}*{Sup Score} & R.F. &0.016&0.026&0.032&0.024&0.027&0.037&0.030&0.041&0.037\\
& T.S. &$-$ & $-$ & $-$ & $-$ & $-$ & $-$ & $-$ & $-$ & $-$\\

\bottomrule
\multicolumn{11}{l}{\emph{Notes}:}\\
\multicolumn{11}{l}{R.F. denotes the rejection frequency.}\\
\multicolumn{11}{p{0.9\textwidth}}{T.S. denotes the rejection frequency where a null hypothesis is rejected only if the robust test of overidentifying restrictions fails to reject the selected IVs' exogeneity.}\\
\multicolumn{11}{l}{Nominal test size: 10\%.}\\
\multicolumn{11}{l}{1,000 Monte Carlo replications.}\\
\end{tabular}
\end{table}

\begin{figure}[H]
    \centering
    	\begin{subfigure}[b]{0.25\textwidth}
    		\includegraphics[width=\textwidth]{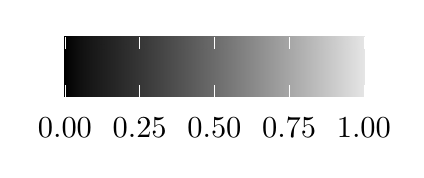}
	    \end{subfigure}
    
       \begin{subfigure}[b]{0.5\textwidth}
    	\centering
        \includegraphics[width=\textwidth]{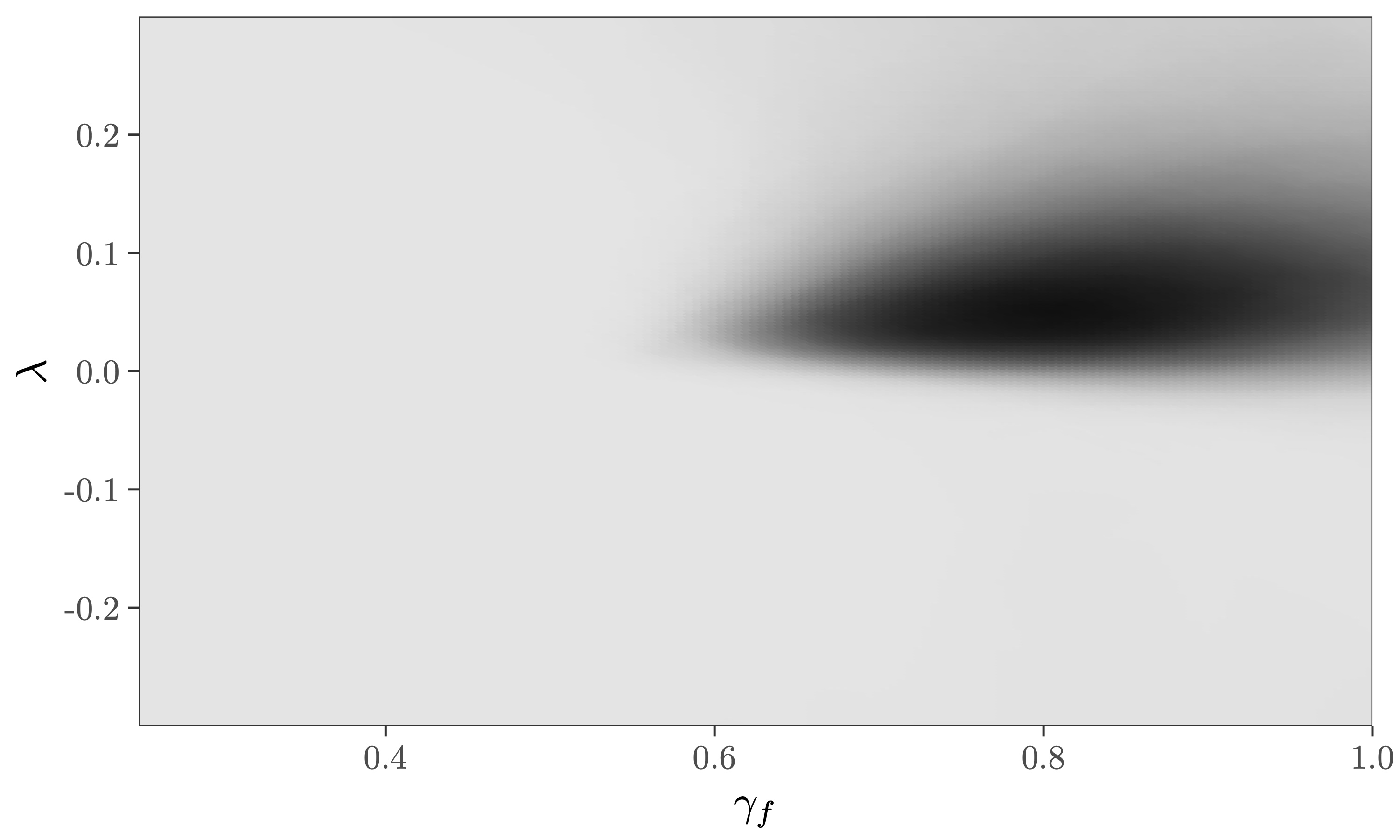}
        \caption{Oracle}
    \end{subfigure}%
    \begin{subfigure}[b]{0.5\textwidth}
    	\centering
        \includegraphics[width=\textwidth]{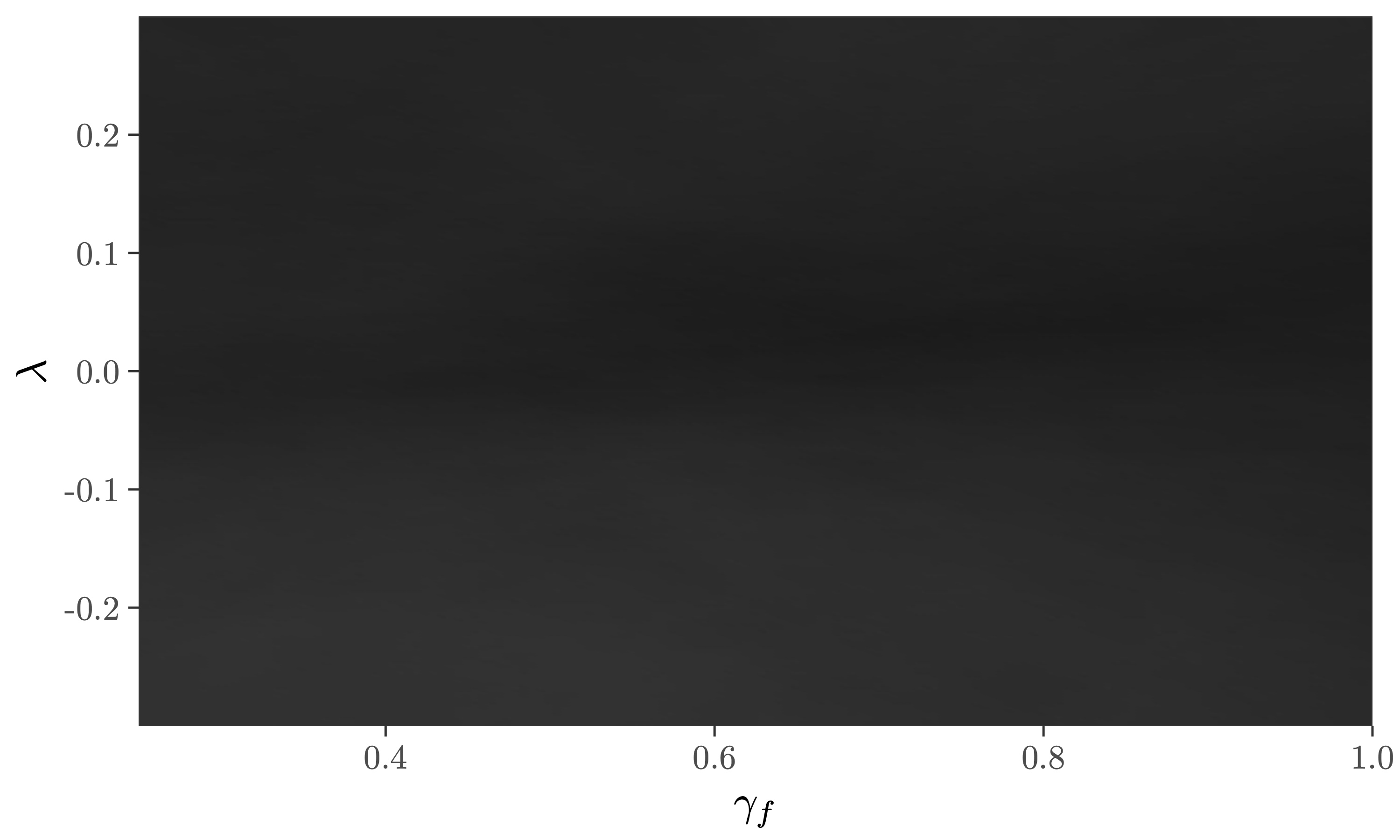}
        \caption{Random}
    \end{subfigure}
    \par\bigskip
    \par\bigskip
  \begin{subfigure}[b]{0.5\textwidth}
    	\centering
        \includegraphics[width=\textwidth]{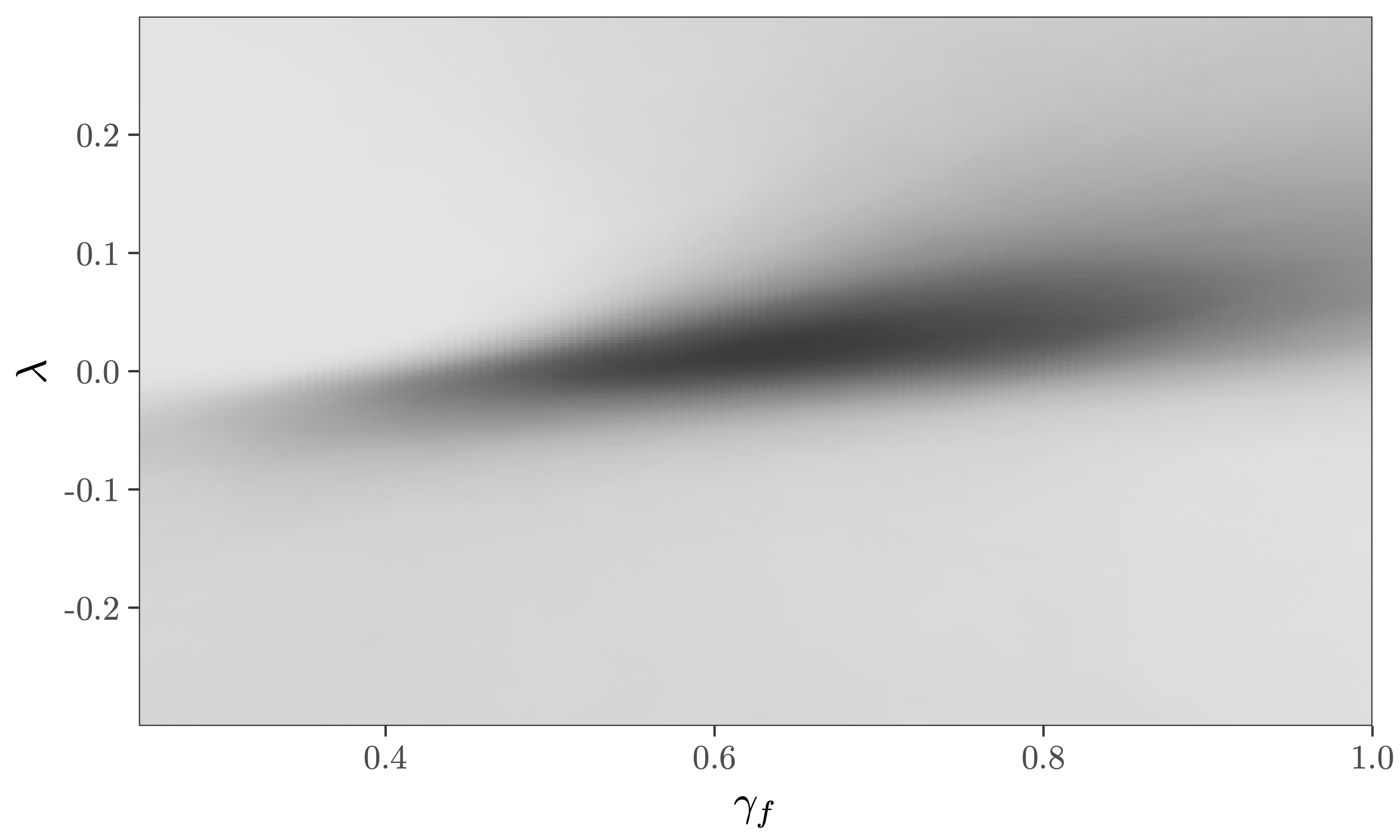}
        \caption{Crude Thresholding}
    \end{subfigure}%
    \begin{subfigure}[b]{0.5\textwidth}
    	\centering
        \includegraphics[width=\textwidth]{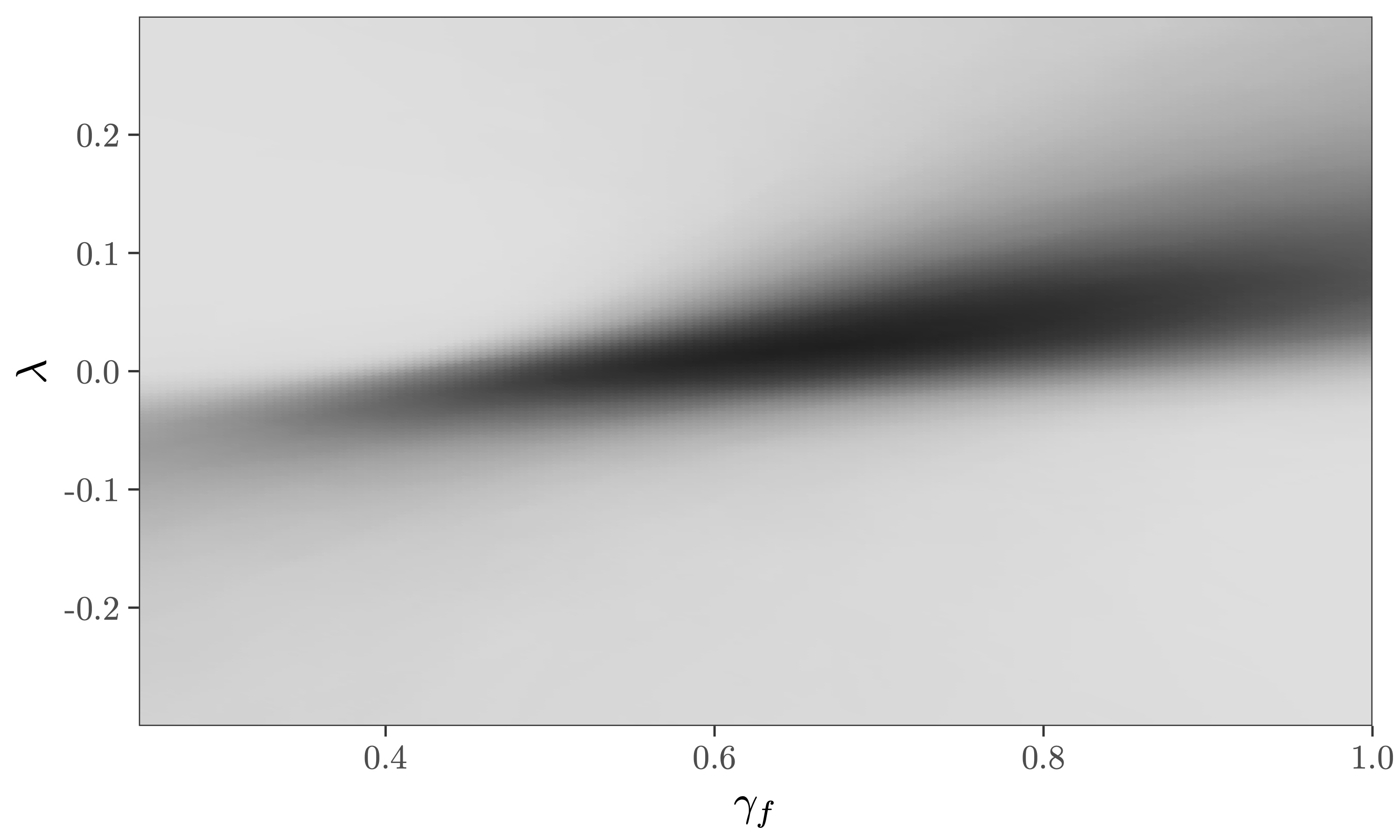}
        \caption{LASSO}
    \end{subfigure}
    \par\bigskip
    \par\bigskip
    \begin{subfigure}[b]{0.5\textwidth}
    	\centering
        \includegraphics[width=\textwidth]{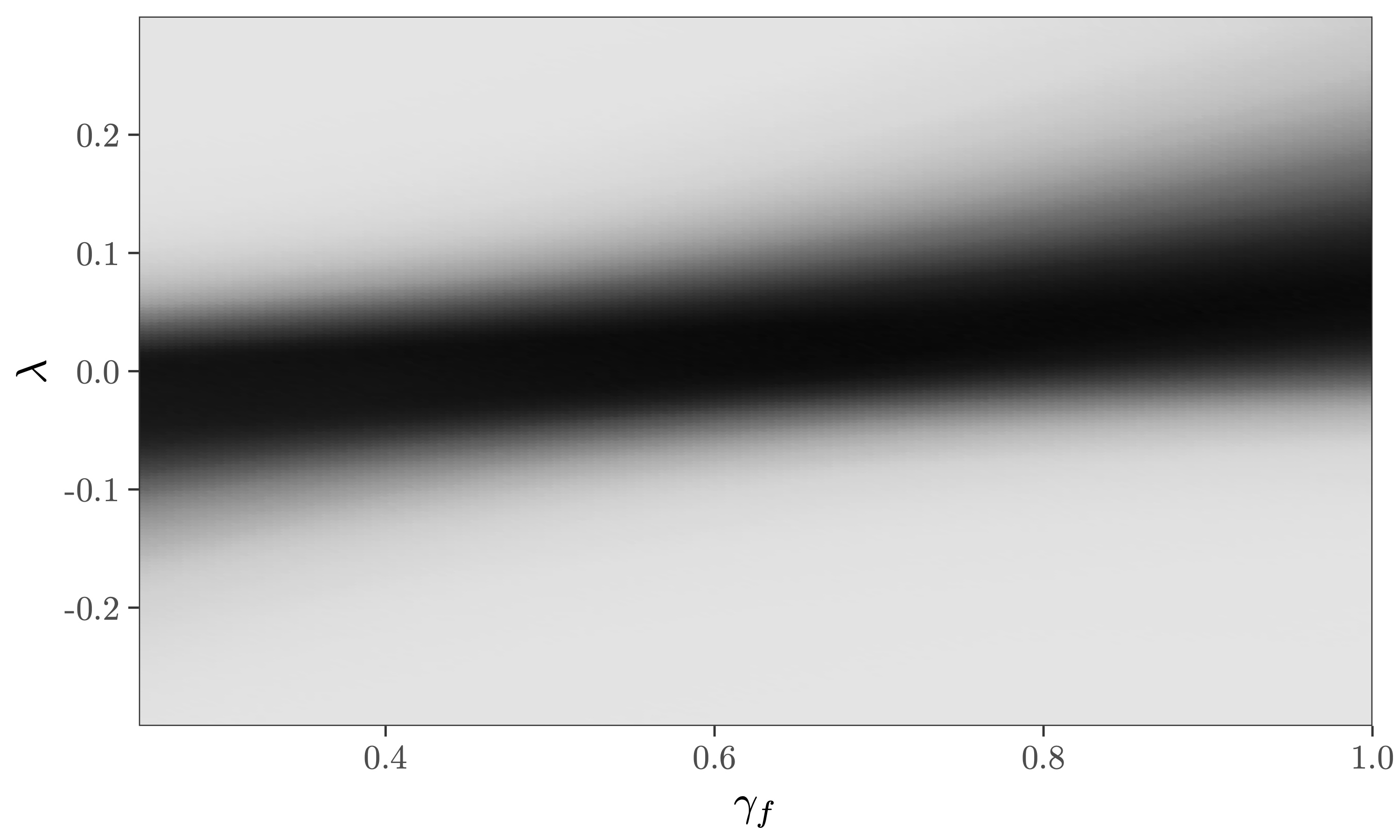}
        \caption{Sup Score}
    \end{subfigure}
	\caption{Simulation results: power. $a_{21} = a_{22} = a_{23} = 0.45$. Nominal test size: 10\%. 1,000 Monte Carlo replications.}\label{Figure-Sim-Results-Power}
\end{figure}

\pagebreak

\section{Empirical Application \label{Section-Empirics}}

Data for the empirical part of this paper is taken from FRED. I use the non-farm labour share as transformed in \citet{Gali:1999tx} as the forcing variable. I use the inflation rate implied by the GDP deflator. I consider the period 1974Q2-2018Q4, and include 90 variables aimed to reflect different parts of the US economy based on the list in \citet{McCracken:2015ct} with four lags each, transforming them as recommended therein.\footnote{I do not include all the variables listed in \citet{McCracken:2015ct} since they are not all available over a sufficiently long period of time.} This yields 179 observations with 359 IVs. Appendix \ref{Appendix-Data} contains a detailed description of the data. 

A natural question to ask is whether considering these 359 IVs is enough to dispel concerns about potential endogeneity biases. Though being more than any number of IVs previously considered in the literature, there are certainly more valid IVs (i.e., additional predetermined variables). However, a substantial endogeneity bias caused by the selection of these 359 IVs would emerge only if the variables were included in the list of \citet{McCracken:2015ct} based on their correlation with the endogenous variables of this application. This seems very unlikely.

For all ad-hoc selection procedures, I limit the number of IVs selected to four, to ensure that overfitting is not a concern, and set the lag-length for the \citet{Newey:1987ua} HAC variance estimator to 4. The results do not change appreciably when other values are chosen.  The confidence sets yielded by the traditional IVs and the ad-hoc selection procedures are shown in Figure \ref{Figure-CS-AD}, and the corresponding IVs are listed in Table \ref{Table-IVs-Selected}. Mirroring the results in Section \ref{Section-Simulation}, the confidence set resulting from random selection is extremely wide, and suggests that the hybrid NKPC is essentially unidentified. The confidence set from applying the LASSO is smaller than the one implied by random selection, but it also does not exclude that the coefficient on expected inflation is in fact equal to zero. 

\begin{table}[H]
\scriptsize
\caption{Identity of the IVs for each of the selection procedures.}\label{Table-IVs-Selected}
\centering
\begin{tabular}{P{0.22\textwidth}P{0.22\textwidth}P{0.22\textwidth}P{0.22\textwidth}}
\toprule
Traditional & Random &  Crude Thresholding & LASSO\\
\midrule

PRS85006173.-1 & WILL5000IND.-3 & PRS85006173.-1 & PRS85006173.-1\\

PRS85006173.-2 & NDMANEMP.-3 & PRS85006173.-2 & PRS85006173.-2\\

GDPDEF.-2 	   & EXUSUK.-3 & DSERRG3M086SBEA.-1&  DSERRG3M086SBEA.-1\\

GDPDEF.-3	   & PERMITMW.-4	& CES3000000008.-1 & SRVPRD.-3\\

\bottomrule
\multicolumn{4}{l}{\tiny\emph{Notes}:}\\
\multicolumn{4}{p{0.88\textwidth}}{\tiny PRS85006173  refers to  Nonfarm Business Sector: Labor Share, GDPDEF refers to  Gross Domestic Product: Implicit Price Deflator, WILL5000IND.-3 refers to  Wilshire 5000 Total Market Index, NDMANEMP.-3 refers to  All Employees, Nondurable Goods, EXUSUK.-3 refers to  U.S. / U.K. Foreign Exchange Rate, PERMITMW.-4 refers to New Private Housing Units Authorized by Building Permits in the Midwest Census Region, DSERRG3M086SBEA refers to  Personal consumption expenditures: Services (chain-type price index), CES3000000008 refers to Average Hourly Earnings of Production and Nonsupervisory Employees, Manufacturing, SRVPRD refers to  All Employees, Service-Providing.}\\
\end{tabular}
\end{table}

\begin{figure}[H]
    \centering
       \begin{subfigure}[b]{0.42\textwidth}
    	\centering
        \includegraphics[width=\textwidth]{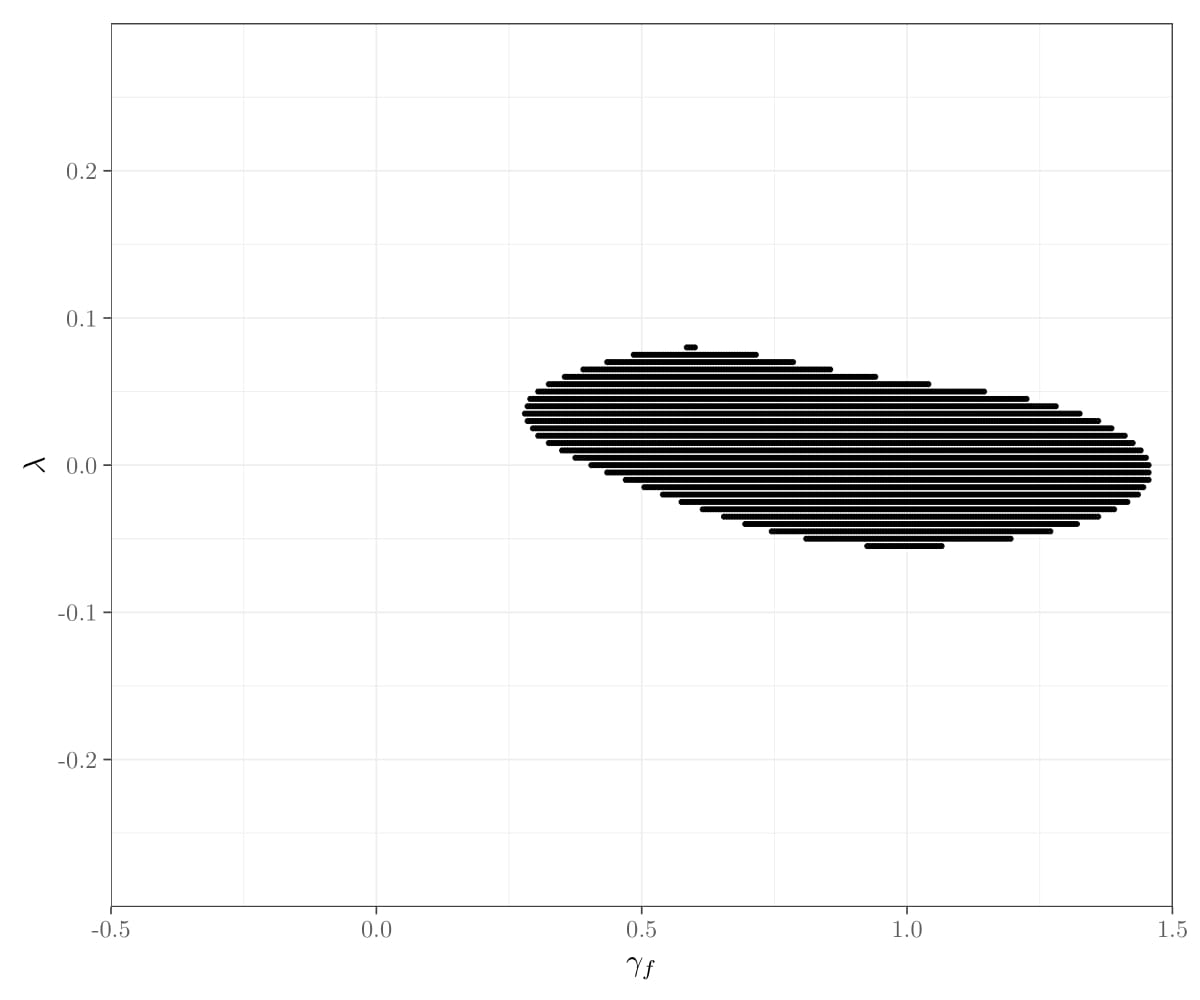}
        \caption{Traditional}
        \label{Figure-CS-Trad}
    \end{subfigure}%
    \begin{subfigure}[b]{0.42\textwidth}
    	\centering
        \includegraphics[width=\textwidth]{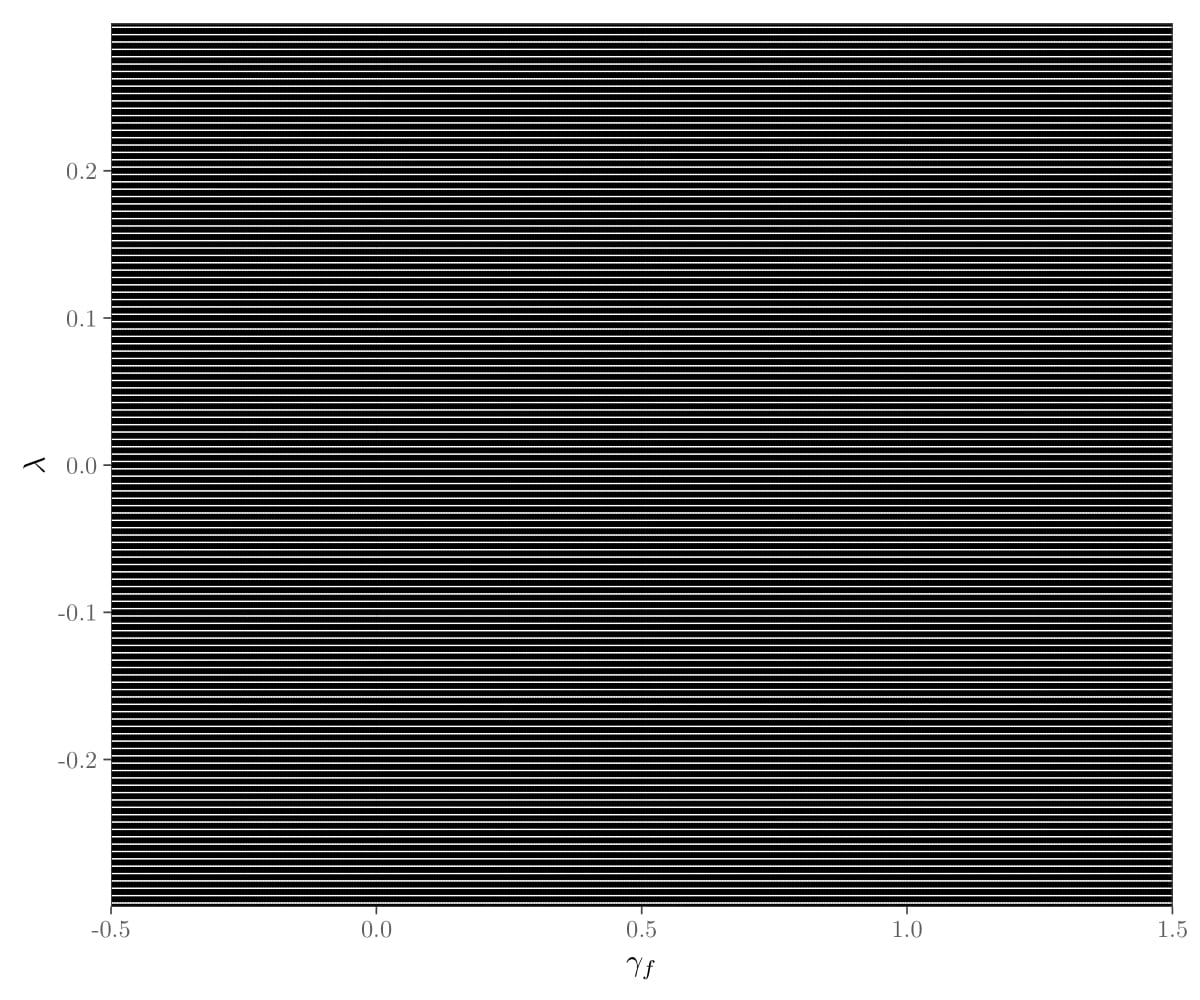}
        \caption{Random}
        \label{Figure-CS-Random}
    \end{subfigure}
  \begin{subfigure}[b]{0.42\textwidth}
    	\centering
        \includegraphics[width=\textwidth]{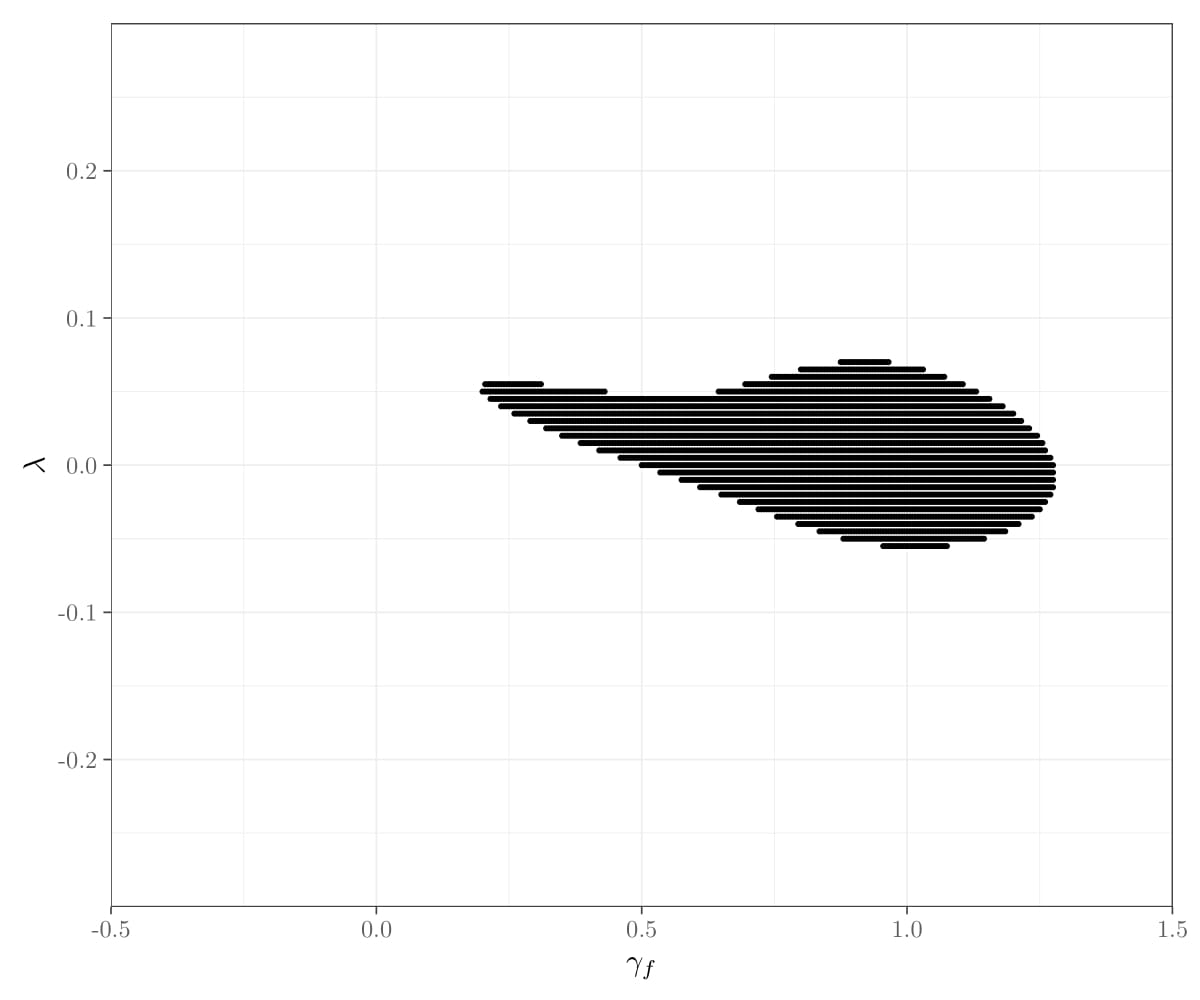}
        \caption{Crude Thresholding}
        \label{Figure-CS-Selected}
    \end{subfigure}%
\begin{subfigure}[b]{0.42\textwidth}
    	\centering
        \includegraphics[width=\textwidth]{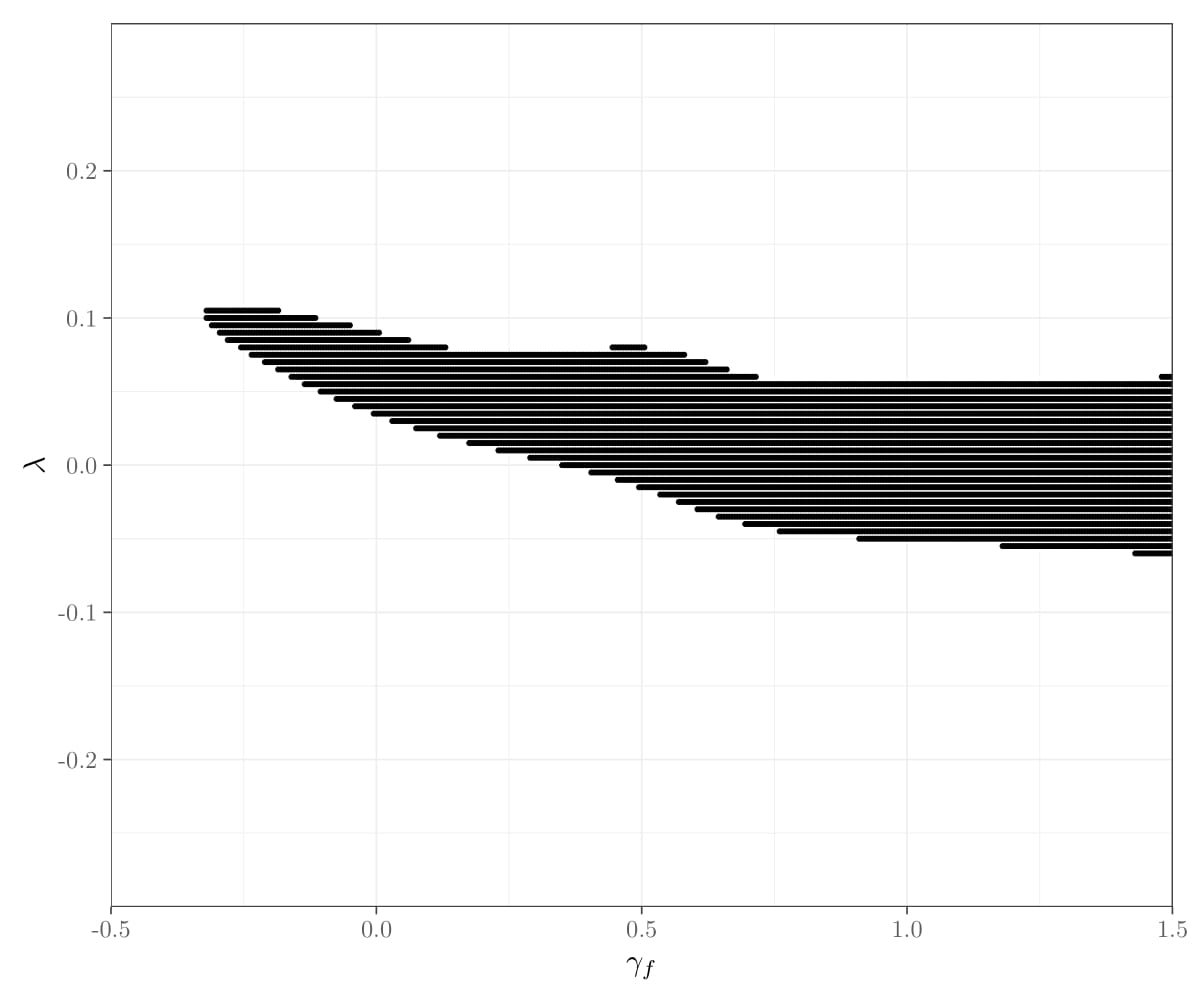}
        \caption{LASSO}
        \label{Figure-CS-LASSO}
    \end{subfigure} 
    \caption{90\% confidence sets for the NKPC in Equation \eqref{Equation-PC-Hybrid} using the $S$ statistics with selected IVs by different ad-hoc procedures. The identity of the IVs selected by each of the procedures is shown in Table \ref{Table-IVs-Selected}.}
    \label{Figure-CS-AD}
\end{figure}

The confidence sets implied by traditional IVs and crude thresholding are qualitatively very similar. This similarity is explained by the fact that the IVs chosen by crude thresholding are very similar to the traditional IVs, as shown in Table \ref{Table-IVs-Selected}. This suggests two things. First, it suggests that the ad-hoc selection procedures used in this paper may in fact provide a reasonable approximation to the approach taken for selecting IVs in the past literature. Second, given that in the simulation exercise in Section \ref{Section-Simulation} crude thresholding yields the worst size distortions of the selection procedures considered, this result suggests that the confidence sets reported in the previous literature are likely to suffer from at least some distortion due to endogeneity bias. In particular, it suggests that the process of trying to find `strong' IVs may have led to an undercovering of the true parameter values. 

The confidence sets of the Sup Score test proposed in this paper for different block lengths (4, 6, 8, and 10) are shown in Figure \ref{Figure-CS}. For all block lengths considered (the results do not appear to be sensitive to the choice of block length), the confidence sets are smaller than the ones yielded by the random selection approach, but wider than for the traditional, crude thresholding, and LASSO approach. This is likely due to a combination of the incorrect size of the latter approaches, and the low power of the Sup Score test documented in Section \ref{Section-Simulation}. The results suggest that while certain parts of the parameter space considered can be rejected at the 10\% level of significance, neither $\lambda$ nor $\gamma_f$ are found to be different from zero for all values of the parameter space considered. 

\begin{figure}[H]
    \centering
       \begin{subfigure}[b]{0.38\textwidth}
    	\centering
        \includegraphics[width=\textwidth]{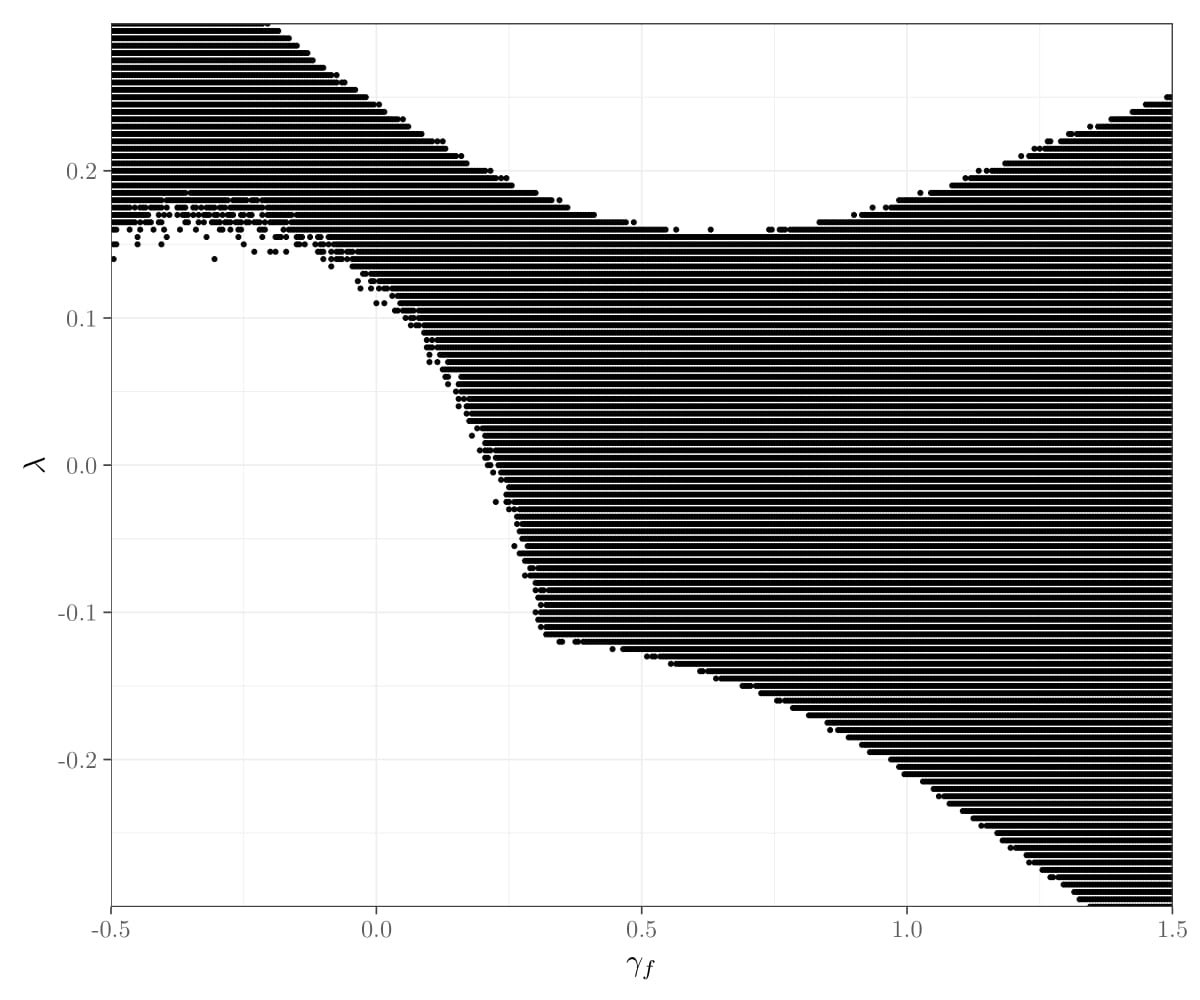}
        \caption{$b_T = 4$}
        \label{Figure-CS-Sup-4}
    \end{subfigure}%
    \begin{subfigure}[b]{0.38\textwidth}
    	\centering
        \includegraphics[width=\textwidth]{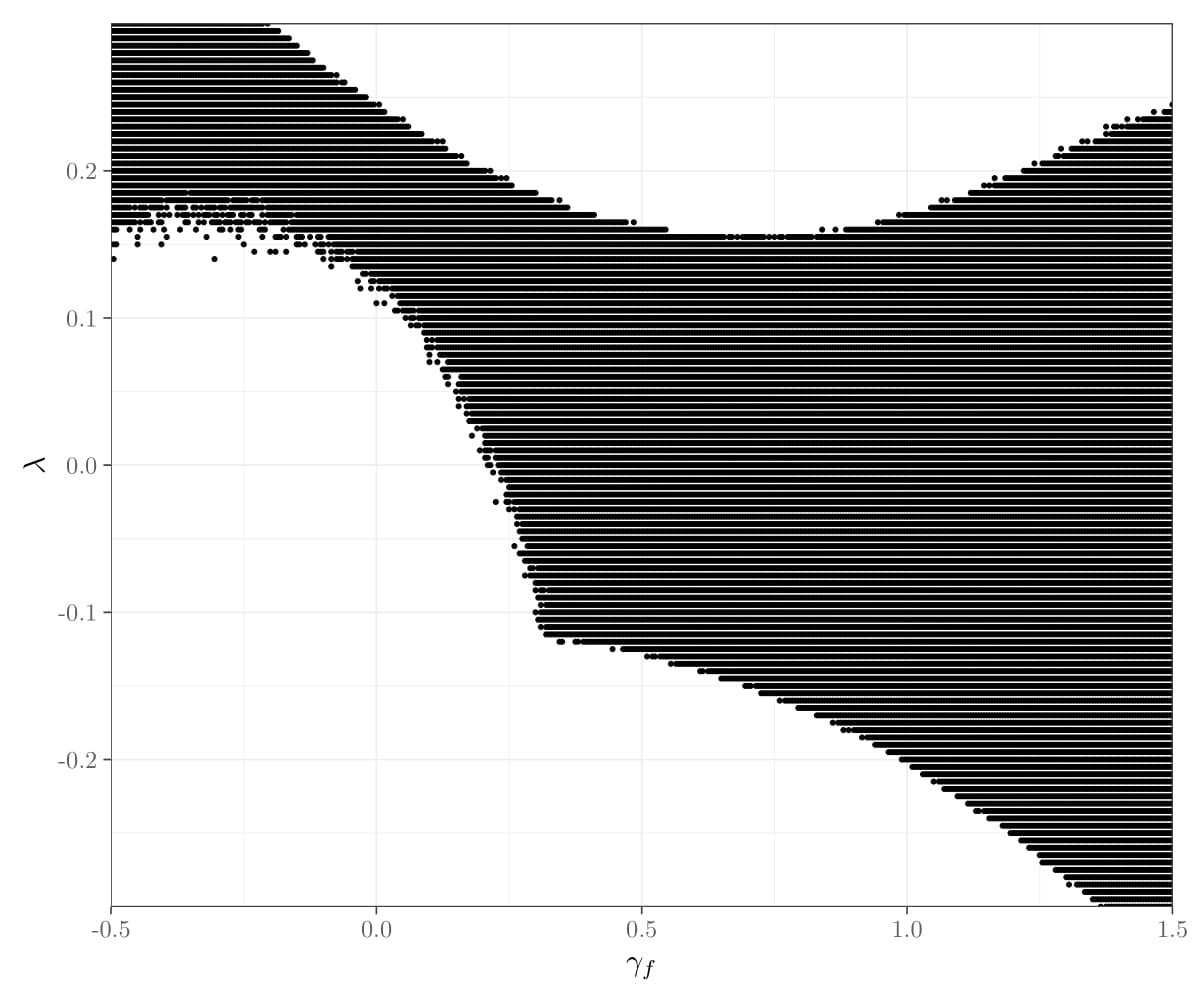}
        \caption{$b_T = 6$}
        \label{Figure-CS-Sup-8}
    \end{subfigure}
  \begin{subfigure}[b]{0.38\textwidth}
    	\centering
        \includegraphics[width=\textwidth]{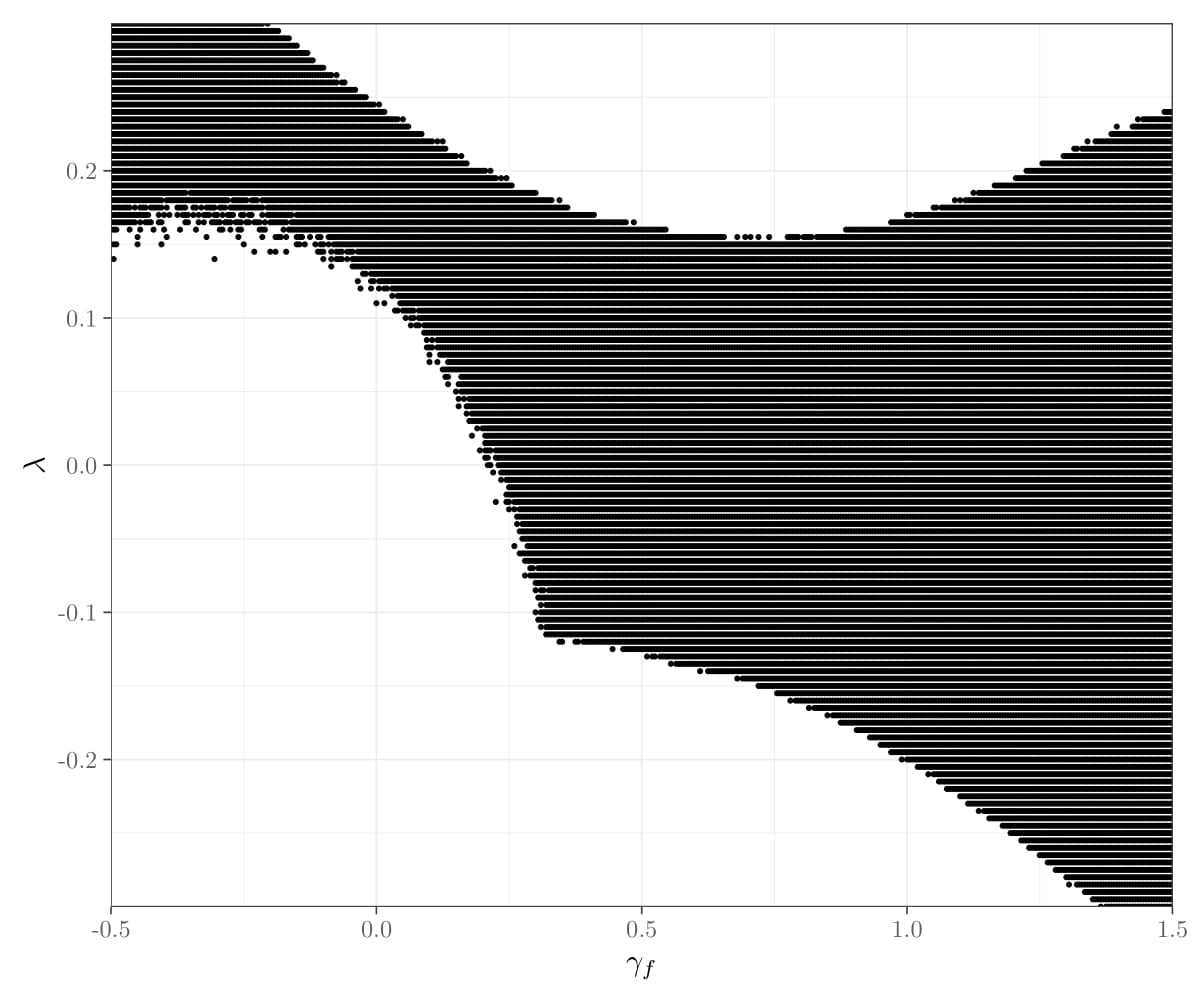}
        \caption{$b_T = 8$}
        \label{Figure-CS-Sup-12}
    \end{subfigure}%
    \begin{subfigure}[b]{0.38\textwidth}
    	\centering
        \includegraphics[width=\textwidth]{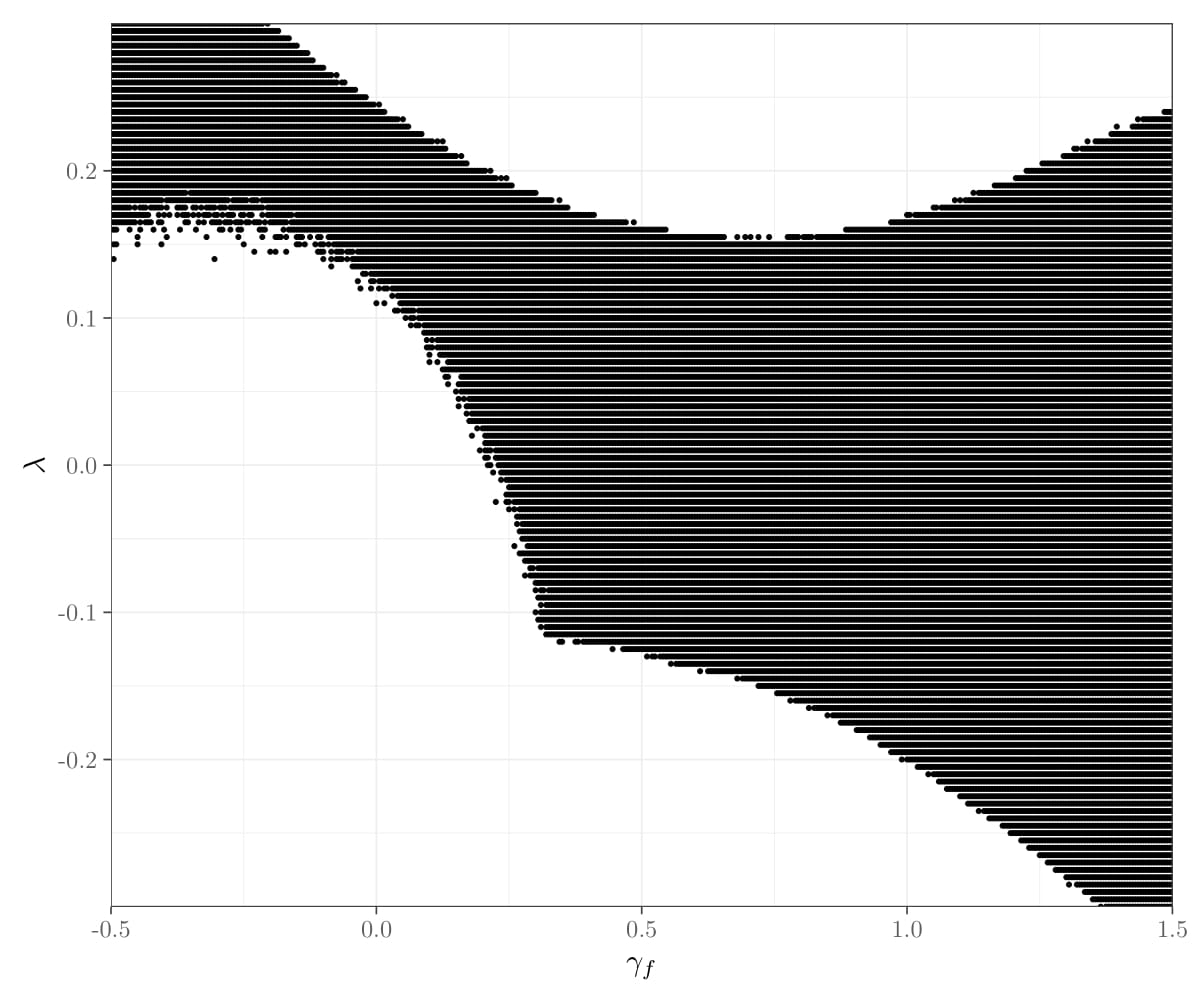}
        \caption{$b_T = 10$}
        \label{Figure-CS-Sup-16}
    \end{subfigure}
    \caption{90\% confidence sets for the NKPC in Equation \eqref{Equation-PC-Hybrid} using the Sup Score test.}
    \label{Figure-CS}
\end{figure}

The Sup Score test can also help shed some light on what the most relevant IVs are, since there is a (likely) unique IV that maximises the Sup Score statistic in Equation \eqref{Equation-TSS} for every null hypothesis being tested. I record the identity of the IV maximising the Sup Score statistic for each null hypothesis tested, and report the results in Table \ref{Table-Sup-Which} and Figure \ref{Figure-Sup-Where}. Table \ref{Table-Sup-Which} shows the identity of the IVs maximising the Sup Score statistic. Figure \ref{Figure-Sup-Where} shows in which part of the parameter space the different IVs maximise the Sup Score statistic. Two things stand out. First, the IVs that feature prominently in Table \ref{Table-IVs-Selected} (i.e., IVs selected by the ad-hoc procedures) also tend to feature in the set of IVs maximising the Sup Score statistic (i.e., lags of PRS85006173 and CES3000000008). The reason why the confidence sets are larger for the Sup Score test is in part due to the Sup Score test being able to account for the very many other valid IVs that these variables were chosen from. Second, there are some IVs that maximise the Sup Score statistic that do not feature in Table \ref{Table-IVs-Selected}, such as the three-period lagged Housing Starts in Northeast Census Region (e.g., HOUSTNE.-3). This relates to the predominant motivation for wanting to consider very many IVs: the truly relevant IVs can often be `exotic', in the sense that intuition alone would not point to their relevance.

\begin{table}[H]
\scriptsize
\caption{Identity of the IVs maximising the Sup Score statistic.}\label{Table-Sup-Which}
\centering
\begin{tabular}{p{0.25\textwidth}p{0.1\textwidth}p{0.55\textwidth}}
\toprule
IV & \# $H_0$ &  Description\\
\midrule

  {CES3000000008.-1} & 13,879 & One-period lag of Average Hourly Earnings of Production and Nonsupervisory Employees, Manufacturing \\

  {HOUSTNE.-3} & 11,021 & Housing Starts in Northeast Census Region \\

  PRS85006173.-3 & 8,404 & Three-period lag of Nonfarm Business Sector: Labor Share \\

  PRS85006173.-1 & 7,263 & One-period lag of Nonfarm Business Sector: Labor Share \\

  PRS85006173.-2 & 4,339 & Two-period lag of Nonfarm Business Sector: Labor Share \\

  CUMFNS.-3 & 1,802 & Three-period lag of Capacity Utilization: Manufacturing\\

  CUSR0000SAS.-2 & 1,742 & Two-period lag of Consumer Price Index for All Urban Consumers: Services in U.S. City Average\\

  IPCONGD.-3 & 68 & Three-period lag of Industrial Production: Consumer Goods\\

  {DDURRG3M086SBEA.-1} & 2 & One-period lagged personal consumption expenditures: durable goods \\

\bottomrule
\end{tabular}
\end{table}

\begin{figure}[H]
	\begin{subfigure}[b]{0.6\textwidth}
    	\centering
        \includegraphics[width=\textwidth]{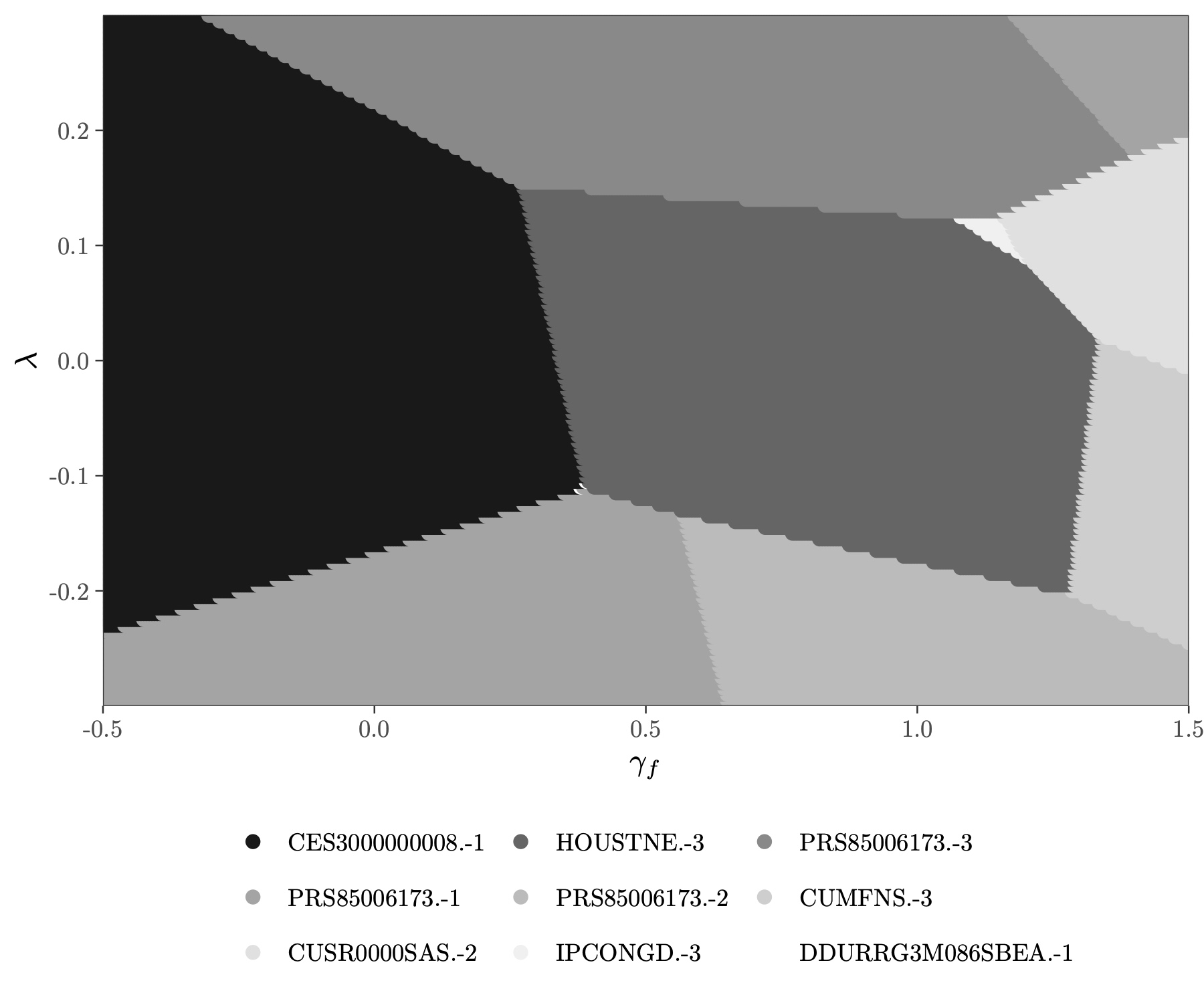}
    \end{subfigure}
    \centering
	\caption{Location in the parameter space $\gamma_f \times \lambda$ where the different IVs in Table \ref{Table-Sup-Which} maximise the Sup Score statistic.}\label{Figure-Sup-Where}

\end{figure}

\section{Conclusion \label{Section-Conclusion}}

IV-based limited-information estimation of single equations has become increasingly popular in Macroeconomics over the last 20 years. Using a simulation exercise based on NKPCs, I showed that selecting IVs in ad-hoc ways (random selection, crude thresholding, and LASSO) can invalidate them, thus yielding invalid inference even if tests with desirable properties (such as robustness to weak identification) are used post-selection. To address this issue, I propose a Sup Score test that remains valid for high-dimensional IVs and for time series data. In the same simulation exercise that showed that ad-hoc selection procedures can lead to invalid inference, this statistic yielded correct size and reasonable power. Finally, I applied the Sup Score test to conduct inference on the US NKPC with 359 IVs on a sample size of 179 observations. The results showed that the confidence sets implied by the Sup Score test are substantially wider than the ones of all other approaches. The simulation results and the empirical application point to the importance of developing further high-dimensional IV methods with good power properties that remain valid under dependence and arbitrarily weak identification.

%% file: USNKPCAppendix.tex
\section{Proof of Example \ref{Example-Z-UGMC} \label{Appendix-Example}}

\begin{proof}

		Assumption \ref{Assumption-Main}.\ref{Assumption-Z} holds by the assumption of mean-zero independent error terms, the assumption of non-zero variances, and the assumption of finite fourth moments.
		
		Assumption \ref{Assumption-Main}.\ref{Assumption-Weird-MC} is satisfied by the assumption of (mean-zero) independent error terms.
	
		Since $\mathbb{E}[\varepsilon_t^4] < \infty$, I can re-write Equation \eqref{Equation-Z-UGMC} as
		\begin{equation*}
			\mathbb{E}[|Z_{tj} - \mathcal{F}_j(\dots, v^*_{-1}, v_0^*, v_1, \dots, v_t)|^4] \mathbb{E}[|\varepsilon_t|^4] \leq \breve{C} {\tilde{\rho}}^t,
		\end{equation*}
		for some new constant $\breve{C}$. Then it follows that
		\begin{equation*}
			\begin{aligned}
				\breve{C} {\tilde{\rho}}^t \geq  \mathbb{E}\left[\left(\mathcal{F}_j(\dots, v_{t-1}, v_t)\varepsilon_t - \mathcal{F}_j(\dots, v^*_{-1}, v^*_0, v_1, \dots, v_t)\varepsilon_t\right)^4\right].
			\end{aligned}
		\end{equation*}
		I now define $\tilde{v}_t = [v_t', \varepsilon_t]'$ (so that $\tilde{v}_t$ is an i.i.d. mean-zero random variable), which yields
		\begin{equation*}
			Z_{tj}\varepsilon_t = \mathcal{F}_j(\dots, v_{t-1}, v_t)\varepsilon_t \equiv \tilde{\mathcal{F}}_j(\dots, \tilde{v}_{t-1}, \tilde{v}_t),
		\end{equation*}
		so that $\mathcal{F}_j$ continues to be a measurable function with arguments that are i.i.d. random variables. Therefore, Assumption \ref{Assumption-Main}.\ref{Assumption-Z-Rep} holds.
		
		Since $\varepsilon_t$ is independent across $t$ and of $v_s$ for $s = 1, \dots, T$ and identically distributed, it follows that
		\begin{equation*}
			\mathcal{F}_j(\dots, v^*_{-1}, v^*_0, v_1, \dots, v_t)\varepsilon_t = \tilde{\mathcal{F}}_j(\dots, \tilde{v}^*_{-1}, \tilde{v}^*_0, \tilde{v}_1, \dots, \tilde{v}_t),
		\end{equation*}
		where $\{\tilde{v}^*_t\}$ are i.i.d. copies of $\{\tilde{v}_t\}$. Thus,
		\begin{equation*}
			\breve{C} {\tilde{\rho}}^t \geq \mathbb{E}\left[\left(\tilde{\mathcal{F}}_j(\dots, \tilde{v}_{t-1}, \tilde{v}_t) - \tilde{\mathcal{F}}_j(\dots, \tilde{v}^*_{-1}, \tilde{v}^*_0, \tilde{v}_1, \dots, \tilde{v}_t)\right)^4\right].
		\end{equation*}
		 Therefore, Assumption \ref{Assumption-Main}.\ref{Assumption-UGMC} holds.
		
\end{proof}

\section{Proof of Theorem \ref{Theorem-Main} \label{Appendix-Proof}}

\begin{proof}
		
	Throughout, it is assumed that the null hypothesis in Equation \eqref{Equation-Null-Hypothesis} holds, so that $\varepsilon_{0t}$ is replaced by $\varepsilon_t$. The proof is a straightforward application of the results in \citet{Zhang:2018hy} (referred to as ZC18 in the sequel) and \citet{Zhang:2014vm} (referred to as ZC14 in the sequel). To this end, let $W_t = [W_{t1}, \dots, W_{tj}]'$ be a Gaussian sequence which is independent of $Z_t\varepsilon_t$ and preserves the autocovariance structure of $Z_t\varepsilon_t$. Let $L_{Z\varepsilon}= \underset{1 \leq j \leq k}{\text{ max }}\frac{1}{\sqrt T}{Z}_j'\varepsilon$ and $L_{W}= \underset{1 \leq j \leq k}{\text{ max }}\frac{1}{\sqrt T}\mathfrak{W}_{j}$ where $\mathfrak{W}_j$ is the $T\times 1$ vector containing the $j$th column of the matrix $W = [W_1, \dots, W_T]$.
		
	I first verify that the conditions in Assumption \ref{Assumption-Main}. are sufficient for Theorem 2.1 in ZC18 to hold. 
	
	Assumption 2.1 in ZC18 holds since by Assumption \ref{Assumption-Main}.\ref{Assumption-Z} $Z_{tj}\varepsilon_t$ has finite fourth moments, so that setting $\mathfrak{D}_n$ in ZC18 to $T^{(3 - 12\tilde{b} - 13b)/32}$, and $h(\cdot)$ in ZC18 to $h(x) = x^4$ satisfies the first of the two possible conditions in Assumption 2.1 of ZC18 by the assumption that $12\tilde{b} + 13b < 3$ (which is implied by the restrictions on $b$ and $\tilde{b}$ given in Assumption \ref{Assumption-Main}.\ref{Assumption-Dim}).
	
	Assumption 2.2 in ZC18 holds by replacing $M$ in ZC18 with $b_T$, and setting $\gamma$ in ZC18 to $\gamma = T^{-(1 - 4\tilde{b} - 7b)/8)} = o(1)$ (see also the sentence immediately following Theorem 3.2 in ZC14).
	
	Assumption 2.3 in ZC18 contains two conditions. The first condition (what they express as $c_1 < \underset{1 \leq j \leq k}{\text{ min }}\sigma_{j, j}\leq \underset{1 \leq j \leq k}{\text{ max }}\sigma_{j, j} < c_2$) holds since by Assumption \ref{Assumption-Main}.\ref{Assumption-Z}, $Z_{tj}\varepsilon_t$ has non-degenerate finite second moments. The second condition (what they express as $\sum_{j = 1}^{+\infty}j\theta_{j, k, 3}< c_3$) is satisfied by Assumption \ref{Assumption-Main}.\ref{Assumption-UGMC}, since, as per Remark 3.2 in \citet{Wang:2019ta}, the GMC condition used in the present paper (and arguably in the literature that uses physical dependence measures more broadly) is equivalent to the one used in ZC18 and ZC14.
		
	Therefore, by Theorem 2.1 in ZC18, under the conditions in Assumption \ref{Assumption-Main}., the process $Z_t\varepsilon_t$ can be approximated by its Gaussian equivalent, i.e.,
	\begin{equation}
		\label{Equation-Bound}
			\underset{a\in \mathbb{R}}{\text{ sup }}\left|\mathbb{P}(L_{Z\varepsilon} \leq a) - \mathbb{P}(L_W \leq a)\right| \lesssim T^{-(1 - 4\tilde{b} - 7b)/8}.
	\end{equation}

	The bound in Equation \eqref{Equation-Bound} satisfies the first condition for Theorem 4.2 in ZC14.\footnote{The careful reader will have noticed that Theorem 4.2 in ZC14 appeals to the conditions in Theorem 3.3 in ZC14 which is virtually the same theorem as Theorem 2.1 in ZC18 except for an additional GMC assumption on the Gaussian equivalent of $Z_t\varepsilon_t$. However, a careful reading of the proof of Theorem 4.2 in ZC14 reveals that this theorem exclusively appeals to the conditions in Theorem 3.3 in ZC14 in order to establish a bound on the Gaussian approximation as in Equation \eqref{Equation-Bound} above. Since Theorem 2.1 in ZC18 establishes this bound without this assumption, the GMC on the Gaussian equivalent of $Z_t\varepsilon_t$ can be dropped in appealing to Theorem 4.2 in ZC14.} It remains to verify Condition 2 of Assumption 4.1 in ZC14. Condition 2 in Assumption 4.1 in ZC14 requires checking four conditions.
	
	The first condition (what they express as $\bar{\sigma}_{x, M}\lor \bar{\sigma}_{x, N}\lesssim n^{s_1}$) is satisfied by Assumption \ref{Assumption-Main}.\ref{Assumption-UGMC}, since by Remark 4.1 in ZC14, the first condition of Condition 2 of Assumption 4.1 in ZC14 is satisfied with $s_1 = 0$ whenever the data in question obeys the GMC condition.

	The second condition (what they express as ${\varsigma}_{x, M}\lor {\varsigma_{x, N}}\lesssim n^{s'_2/2}$) is satisfied by setting their $M, N$ to $b_T$ and $s'_2$ to $\tilde{b}$ and noticing that for all $j = 1, \dots, k$, 
	\begin{equation*}
		\left(\frac{1}{b_T^2}\mathbb{E}\left[\left|\sum_{t = 1}^{b_T}{Z_{tj}\varepsilon_t}\right|^4\right]\right)^{1/4} \leq \left(\frac{1}{b_T^2}\mathbb{E}\left[\underset{1 \leq t \leq b_T}{\text{ max }}Z_{tj}^4\varepsilon_t^4\right]{b_T^4}\right)^{1/4} \leq b_T^{1/2}\mathbb{E}\left[\underset{1 \leq t \leq T}{\text{ max }}Z_{tj}^4\varepsilon_t^4\right] \lesssim b_T^{1/2}
	\end{equation*}
	by Assumption \ref{Assumption-Main}.\ref{Assumption-Z}. Thus,
	\begin{equation*}
		\left(\mathbb{E}\left[\underset{1 \leq j \leq k}{\text{ max }}\left|\sum_{t = 1}^{b_T}\frac{Z_{tj}\varepsilon_t}{\sqrt{b_T}}\right|^4\right]\right)^{1/4} \lesssim b_T^{1/2} \lesssim T^{\tilde{b}/2},
	\end{equation*}
	by Assumption \ref{Assumption-Main}.\ref{Assumption-Dim}. This ensures that the second condition of Condition 2 of Assumption 4.1 in ZC14 is satisfied.

	The third condition (what they express as ${\varpi}_{x}\lesssim n^{s_3}$) is satisfied by Assumption \ref{Assumption-Main}.\ref{Assumption-Weird-MC} and setting $s_3 = \breve{b}$.
	
	The fourth condition (what they express as $s_b' > 0$) is satisfied since $(1-6\tilde{b})/2 > 0$, $(1-6b-\tilde{b})/2 - \tilde{b} > 0$, and $\tilde{b} - 2b - \breve{b} > 0$ by the assumption made on $b$, $\tilde{b}$, and $\breve{b}$ in Assumption \ref{Assumption-Main}.\ref{Assumption-Dim} and Assumption \ref{Assumption-Main}.\ref{Assumption-Weird-MC}.
	
	It is hence possible to invoke Theorem 4.2 in ZC14, which yields
	\begin{equation*}
		\underset{\alpha \in (0, 1)}{\text{ sup }}\left|\mathbb{P}(L_{Z\varepsilon}\leq \tilde{c}(\alpha)) - \alpha \right| \lesssim T^{-c},
	\end{equation*}
	where
	 \begin{equation*}
	 	\begin{aligned}
	 	\tilde{c}(\alpha) &= \text{inf}\left\{\gamma \in \mathbb{R}:\mathbb{P}(\tilde{L}_{\hat A} \leq \gamma| \{{Z}_t{\varepsilon}_{0t}\}_{t = 1}^T) \geq 1-\alpha\right\},\\
		\tilde{L}_{\hat A} &= \underset{1 \leq j \leq k}{\text{ max }} \frac{1}{\sqrt T}\sum_{t = 1}^{l_T}\hat{A}_{tj}e_t,
		\end{aligned}
	\end{equation*}	
and $\{e_t\}$ is a sequence of i.i.d. $\mathcal{N}[0, 1]$ random  variables. The constant $c$ is positive, since $(1- 5b - \tilde{b})/2 > 0$, $(1-6b-\tilde{b})/2 - \tilde{b} > 0$, $\tilde{b} - 2b - \breve{b} > 0$, and $(1 - 4\tilde{b} - 7b) > 0$ by the assumption made on $b$, $\tilde{b}$, and $\breve{b}$ in Assumption \ref{Assumption-Main}.\ref{Assumption-Dim} and Assumption \ref{Assumption-Main}.\ref{Assumption-Weird-MC}.

Finally, notice that the procedure proposed in this paper is computing only the means of random variables. This means that the `influence function' ($IF$ in ZC14) does not have to be estimated (since the true value is known under the null hypothesis). It also means that the statistic is `exactly linear', i.e., the remainder term $\mathcal{R}_{N_0}$ in ZC14 is zero. This implies that the two conditions in Assumption 5.1 in ZC14 are trivially satisfied (since, in their notation, $\mathcal{E}_{AB} = \mathcal{R}_{N_0} = 0$). Also, the block length of $Z_{tj}\varepsilon_t$ is simply unity so that $N_0$ in ZC14 is simply $T$ and the dimension of the parameter to be estimated ($q_0$ in their notation) is simply the number of IVs considered, $k$. By Theorem 5.1 in ZC14, which requires the conditions for Theorem 4.1 and Assumption 5.1 in ZC14 to hold, and the identifying moment condition $\mathbb{E}[Z'\varepsilon] = 0$, it hence follows that
	\begin{equation*}
		\underset{\alpha \in (0, 1)}{\text{ sup }}\left|\mathbb{P}\left(\underset{1 \leq j \leq k}{\text{ max }}\sqrt T\left|\frac 1 T Z_j'\varepsilon\right|\leq {c}(\alpha)\right) - \alpha \right| \lesssim T^{-c},
	\end{equation*}
i.e.,
	\begin{equation*}
		\underset{\alpha \in (0, 1)}{\text{ sup }}\left|\mathbb{P}\left(\mathcal{R}\leq {c}(\alpha)\right) - \alpha \right| \lesssim T^{-c}.
	\end{equation*}
Letting $T\to \infty$ yields the required result.

\end{proof}

\section{Derivation of Concentration Parameters for Simulations \label{Appendix-Conc-Par}}

I first derive the concentration parameter for the unobserved oracle first stage. The derivations for this concentration parameter are very similar to those in the online appendix of \citet{Mavroeidis:2014ge}.

The endogenous variables in the model the econometrician estimates (Equation \eqref{Equation-Estimate-Econometrician}) can be written in terms of the excluded IVs as\footnote{As in \citet[Online Appendix]{Mavroeidis:2014ge}, the constant can be omitted for the purposes of deriving the concentration matrix because in all simulations it is set equal to zero in the structural NKPC.}
\begin{equation}
	\label{Equation-Raw-First-Stage}
	\begin{bmatrix}
	\pi_{t+1} - \pi_{t-1}\\
	s_t	
	\end{bmatrix} =
	\underbrace{\begin{bmatrix}
		1 & 0 & 0\\
		0 & 0 & 0\\
	\end{bmatrix}}_{E_1}\begin{bmatrix}
		\pi_{t+1}\\
		s_{t+1}\\
		f_{t+1}\\
	\end{bmatrix} + 
	\underbrace{\begin{bmatrix}
		0 & 0 & 0\\
		0 & 1 & 0\\
	\end{bmatrix}}_{E_2}\begin{bmatrix}
		\pi_t\\
		s_t\\
		f_t
	\end{bmatrix} + \underbrace{\begin{bmatrix}
		-1 & 0 & 0\\
		0 & 0 & 0
	\end{bmatrix}}_{E_3}\begin{bmatrix}
		\pi_{t-1}\\
		s_{t-1}\\
		f_{t-1}
	\end{bmatrix}.
\end{equation}

Define $p_t = [\pi_{t+1} - \pi_{t-1}, s_t]'$, $R_t = [\pi_t, s_t, f_t]'$, $u_t = [u_{1t}, u_{2t}, u_{3t}]'$,
\begin{equation*}
\Psi = \begin{bmatrix}
		a_{11} & a_{12} & a_{13}\\
		a_{21} & a_{22} & a_{23}\\
		a_{31} & a_{32} & a_{33}\\
	\end{bmatrix},\text{ and } \Omega = \begin{bmatrix}
 	\omega_{11} & \omega_{12} & \omega_{13}\\
 	\omega_{21} & \omega_{22} & \omega_{23}\\
 	\omega_{31} & \omega_{32} & \omega_{33}\\
 \end{bmatrix}
 .
\end{equation*}

Equation \eqref{Equation-Raw-First-Stage} can now be written as
\begin{equation*}
	\begin{aligned}
	p_t &= E_1R_{t+1} + E_2R_t + E_3R_{t-1}\\	
	    &= (E_1\Psi^2 + E_2\Psi + E_3)R_{t-1} + (E_1\Psi + E_2)u_t + E_1u_{t+1}\\
	    &= DR_{t-1} + w_t,
	\end{aligned}
\end{equation*}
for $D = E_1\Psi^2 + E_2\Psi + E_3$, $w_t = (E_1\Psi + E_2)u_t + E_1u_{t+1}$.

Assuming that $R_t$ is stationary, and letting $\Gamma = \mathbb{V}\left[R_t\right]$,
\begin{equation*}
	\text{vec}(\Gamma) = \left(I_9 - \Psi\otimes \Psi\right)^{-1}\text{vec}(\Omega).
\end{equation*}

The population projection of $p_t$ on $R_{t-1}$ has coefficient matrix given by
\begin{equation*}
	\begin{aligned}
	M &= \mathbb{E}[p_tR_{t-1}']\Gamma^{-1}\\
	  &= \mathbb{E}\left[(DR_{t-1} + w_t)R_{t-1}'\right]\Gamma^{-1}\\
	  & = D\mathbb{E}\left[R_{t-1}R_{t-1}'\right]\Gamma^{-1}\\
	  & = D,
	\end{aligned}
\end{equation*}
since $\mathbb{E}[w_tR_{t-1}'] = \mathbb{E}\left[\left((E_1\Psi + E_2)u_t + E_1u_{t+1}\right)R_{t-1}'\right] = 0$.

The projection error of the unobserved oracle first stage is given by
\begin{equation*}
	\begin{aligned}
	e_t &= p_t - MR_{t-1}\\
		&= DR_{t-1} + w_t - DR_{t-1}\\
		&= w_t.
	\end{aligned}
\end{equation*}

The variance of the population projection error of the unobserved oracle first stage $\Sigma = \mathbb{V}[e_t]$ is hence given by
\begin{equation*}
	\begin{aligned}
	\Sigma &= \mathbb{V}[w_t]\\
		   &= (E_1\Psi + E_2)\Omega(E_1A + E_2)' + E_1\Omega E_1'.\\
	\end{aligned}
\end{equation*}

The concentration matrix of the unobserved oracle first stage is then given by
\begin{equation*}
	C = T\Sigma^{-1/2}D\Gamma D'\Sigma^{-1/2 '},
\end{equation*}
where $\Sigma^{-1/2}\Sigma^{-1/2, '} = \Sigma^{-1}$. The minimum eigenvalue $\mu^2_O$ of matrix $C$ gives the concentration parameter of the unobserved oracle first stage.

The steps to derive the concentration matrix corresponding to the first stage observed by the econometrician are similar. Define $\tilde{R}_t = [\pi_{t}, s_{t}, Q_t']'$, and
\begin{equation*}
	\Xi = \begin{bmatrix}
 1 & 0 & 0\\
 0 & 1 & 0\\
 0_{m\times 1} & 0_{m\times 1} & \xi	
 \end{bmatrix}, F =\begin{bmatrix}
 0_{2\times 2} & 0_{2\times m}\\
 0_{m\times 2} & I_m
\end{bmatrix},
\end{equation*}
so that $\tilde{R} = \Xi R_t + \tilde{u}_t$, where $\tilde{u}_t = [0, 0, u_{4t}']'$, $\mathbb V[\tilde{u}_t] = F$.

Assuming that $R_t$ is stationary, $\tilde{R}_t$ is also stationary. Letting $\tilde{\Gamma} = \mathbb{V}[\tilde{R}_t]$,
\begin{equation*}
	\begin{aligned}
	\tilde \Gamma &= \mathbb{V}\left[\Xi R_t + \tilde{u}_t\right]\\
				  &= \Xi\mathbb{V}\left[R_t\right]\Xi' + F\\
				  &= \Xi\Gamma\Xi' + F.
	\end{aligned}
\end{equation*}

The population projection of $p_t$ on $\tilde{R}_{t-1}$ has coefficient matrix given by
\begin{equation*}
	\begin{aligned}
	\tilde M &= \mathbb{E}[p_t\tilde{R}_{t-1}']{\tilde{\Gamma}}^{-1}\\
			 &= \mathbb{E}[p_t({R}_{t-1}'\Xi' + \tilde{u}_t')]{\tilde{\Gamma}}^{-1}\\
			 &= \mathbb{E}[p_t{R}_{t-1}']\Xi'{\tilde \Gamma}^{-1	}  + \mathbb{E}[p_t\tilde{u}_t']{\tilde{\Gamma}}^{-1}\\
			 &= \mathbb{E}[(DR_{t-1} + w_t){R}_{t-1}']\Xi'{\tilde{\Gamma}}^{-1}\\
			 & = D\mathbb{E}[R_{t-1}R_{t-1}']\Xi'{\tilde{\Gamma}}^{-1} + \mathbb{E}[w_tR_{t-1}']\Xi'{\tilde{\Gamma}}^{-1}\\
			 &= D\Gamma\Xi'{\tilde{\Gamma}}^{-1}.\\
	\end{aligned}
\end{equation*}

The projection error for the observed first stage is given by
\begin{equation*}
	\begin{aligned}
		\tilde{e}_t &= p_t - \tilde{M}\tilde{R}_{t-1}\\
					&=DR_{t-1} + w_t - \tilde{M}\tilde{R}_{t-1}\\
					&= DR_{t-1} + w_t - \tilde{M}(\Xi R_{t-1} + \tilde{u}_{t-1})\\
					&= (D - \tilde{M}\Xi)R_{t-1} + w_t - M\tilde{u}_{t-1}.
	\end{aligned}
\end{equation*}

The variance of the population projection error of the observed first stage $\tilde{\Sigma} = \mathbb{V}[\tilde{e}_t]$ is hence given by
\begin{equation*}
	\begin{aligned}
	\tilde{\Sigma} &= \mathbb{V}[(D - \tilde{M}\Xi)R_{t-1}] + \mathbb{V}[w_t] + \mathbb{V}[\tilde{M}\tilde{u}_{t-1}]\\
				   &= (D- \tilde{M}\Xi)\Gamma(D - \tilde{M}\Xi)' + \Sigma + \tilde{M}F\tilde{M}'.
	\end{aligned}
\end{equation*}

The concentration matrix of the unobserved oracle first stage is then given by
\begin{equation*}
	\tilde{C} = T\tilde{\Sigma}^{-1/2}\tilde{M}\tilde{\Gamma}\tilde{M}'\tilde{\Sigma}^{-1/2 '},
\end{equation*}
where $\tilde{\Sigma}^{-1/2}\tilde{\Sigma}^{-1/2, '} = \tilde{\Sigma}^{-1}$. The minimum eigenvalue $\mu^2_E$ of matrix $\tilde{C}$ gives the concentration parameter of the observed first stage.

\newgeometry{left=1cm,bottom=0.1cm}

\section{Data \label{Appendix-Data}}

\begin{table}[H]
\tiny
\caption{Data description.}\label{Table-Data}
\centering
\begin{tabular}{lp{0.3\textwidth}l lp{0.4\textwidth}l}
\toprule
Code  &  Description & Type & Code & Description & Type\\
\midrule

RPI & Real Personal Income  & G  & PERMIT & New Private Housing Units Authorized by Building Permits &  G\\ 
  INDPRO & Industrial Production Index  & G  & PERMITNE &New Private Housing Units Authorized by Building Permits in the Northeast Census Region  &  G\\ 
  CUMFNS & Capacity Utilization: Manufacturing  & D  & PERMITMW & New Private Housing Units Authorized by Building Permits in the Midwest Census Region & G \\ 
  IPFINAL & Industrial Production: Final Products (Market Group)  & G & PERMITS & New Private Housing Units Authorized by Building Permits in the South Census Region & G \\ 
  IPCONGD & Industrial Production: Consumer Goods  & G & PERMITW & New Private Housing Units Authorized by Building Permits in the West Census Region &  G\\ 
  IPDCONGD & Industrial Production: Durable Consumer Goods & G & DPCERA3M086SBEA & Real personal consumption expenditures (chain-type quantity index)  &  G\\ 
  IPNCONGD & Industrial Production: Nondurable Consumer Goods & G & CMRMTSPL & Real Manufacturing and Trade Industries Sales  &  G\\ 
  IPBUSEQ & Industrial Production: Business Equipment & G & UMCSENT & University of Michigan: Consumer Sentiment &  G\\ 
  IPMAT & Industrial Production: Materials & G & M1SL & M1 Money Stock &  G\\ 
  IPMANSICS & Industrial Production: Manufacturing (SIC) & G & M2SL & M2 Money Stock &  G\\ 
  IPB51222s & Industrial Production: Residential utilities & G & TOTRESNS & Total Reserves of Depository Institutions &  G\\ 
  IPFUELS & Industrial Production: Fuels  & G & BUSLOANS & Commercial and Industrial Loans, All Commercial Banks &  G\\ 
  CLF16OV & Civilian Labor Force Level & G & REALLN & Real Estate Loans, All Commercial Banks &  G\\ 
  UNRATE & Unemployment Rate & N & NONREVSL & Total Nonrevolving Credit Owned and Securitized, Outstanding & G \\ 
  UEMPMEAN & Average Weeks Unemployed & D & DTCOLNVHFNM &Consumer Motor Vehicle Loans Owned by Finance Companies, Outstanding  &  G\\ 
  UEMPLT5 & Number Unemployed for Less Than 5 Weeks &  G& DTCTHFNM & Total Consumer Loans and Leases Owned and Securitized by Finance Companies, Outstanding &  G\\ 
  UEMP5TO14 & Number Unemployed for 5-14 Weeks & G & INVEST & Securities in Bank Credit, All Commercial Banks & G \\ 
  UEMP15OV & Number Unemployed for 15 Weeks \& Over & G & FEDFUNDS & Effective Federal Funds Rate &  G\\ 
  UEMP15T26 & Number Unemployed for 15-26 Weeks  & G & TB3SMFFM & 3-Month Treasury Bill Minus Federal Funds Rate &  N\\ 
  UEMP27OV & Number Unemployed for 27 Weeks \& Over & G & TB6SMFFM & 6-Month Treasury Bill Minus Federal Funds Rate & N \\ 
  PAYEMS & All Employees, Total Nonfarm & G & T1YFFM & 1-Year Treasury Constant Maturity Minus Federal Funds Rate & N \\ 
  USGOOD & All Employees, Goods-Producing & G & T5YFFM & 5-Year Treasury Constant Maturity Minus Federal Funds Rate & N \\ 
  CES1021000001 & All Employees, Mining & G & T10YFFM & 10-Year Treasury Constant Maturity Minus Federal Funds Rate & N \\ 
  USCONS & All Employees, Construction & G & AAAFFM & Moody's Seasoned Aaa Corporate Bond Minus Federal Funds Rate &  N\\ 
  MANEMP & All Employees, Manufacturing & G & BAAFFM & Moody's Seasoned Baa Corporate Bond Minus Federal Funds Rate & N \\ 
  PRS85006173 & Nonfarm Business Sector: Labor Share & GG & TWEXMMTH & Trade Weighted U.S. Dollar Index: Major Currencies, Goods &  G\\ 
  DMANEMP & All Employees, Durable Goods & G & EXSZUS & Switzerland / U.S. Foreign Exchange Rate & G \\ 
  NDMANEMP & All Employees, Nondurable Goods & G & EXJPUS & Japan / U.S. Foreign Exchange Rate &  G\\ 
  SRVPRD & All Employees, Service-Providing & G & EXUSUK & U.S. / U.K. Foreign Exchange Rate & G \\ 
  USTPU & All Employees, Trade, Transportation, and Utilities & G & EXCAUS & Canada / U.S. Foreign Exchange Rate & G \\ 
  USWTRADE & All Employees, Wholesale Trade & G & WPSFD49502 & Producer Price Index by Commodity for Final Demand: Personal Consumption Goods  & G \\ 
  USTRADE & All Employees, Retail Trade &  G & WPSID61 & Producer Price Index by Commodity for Intermediate Demand by Commodity Type: Processed Goods for Intermediate Demand &  G\\ 
  USFIRE & All Employees, Financial Activities & G & WTISPLC & Spot Crude Oil Price: West Texas Inter mediate (WTI) & G \\ 
  USGOVT & All Employees, Government & G & PPICMM & Producer Price Index by Commodity Metals and metal products: Primary nonferrous metals & G \\ 
  CES0600000007 & Average Weekly Hours of Production and Nonsupervisory Employees, Goods-Producing & G & CPIAUCSL & Consumer Price Index for All Urban Consumers: All Items in U.S. City Average &  G\\ 
  AWOTMAN & Average Weekly Overtime Hours of Production and Nonsupervisory Employees, Manufacturing & D & CUSR0000SAC & Consumer Price Index for All Urban Consumers: Commodities in U.S. City Average & G \\ 
  AWHMAN & Average Weekly Hours of Production and Nonsupervisory Employees, Manufacturing & D & CUSR0000SAD & Consumer Price Index for All Urban Consumers: Durables in U.S. City Average &  G\\ 
  CES0600000008 & Average Hourly Earnings of Production and Nonsupervisory Employees, Goods-Producing & G & CUSR0000SAS & Consumer Price Index for All Urban Consumers: Services in U.S. City Average &  G\\ 
  CES2000000008 & Average Hourly Earnings of Production and Nonsupervisory Employees, Construction & G & PCEPI & Personal Consumption Expenditures: Chain-type Price Index  &  G\\ 
  CES3000000008 & Average Hourly Earnings of Production and Nonsupervisory Employees, Manufacturing & G & DDURRG3M086SBEA & Personal consumption expenditures: Durable goods (chain-type price index) &  G\\ 
  HOUST & Housing Starts: Total: New Privately Owned Housing Units Started &  G& DNDGRG3M086SBEA & Personal consumption expenditures: Nondurable goods (chain-type price index)  &  G\\ 
  HOUSTNE & Housing Starts in Northeast Census Region & G & DSERRG3M086SBEA & Personal consumption expenditures: Services (chain-type price index)  &  G\\ 
  HOUSTMW & Housing Starts in Midwest Census Region & G & GDPDEF & Gross Domestic Product: Implicit Price Deflator &  G\\ 
  HOUSTS & Housing Starts in South Census Region & G & WILL5000IND & Wilshire 5000 Total Market Index & G \\ 
  HOUSTW & Housing Starts in West Census Region &  G & GDPC1 & Real Gross Domestic Product & G \\ 
\bottomrule
\multicolumn{6}{l}{\emph{Notes}:}\\
\multicolumn{6}{l}{N refers to no transformation of the data.}\\
\multicolumn{6}{l}{G refers transforming variable $f_t$ $f_t$ by computing $100\left(\log(f_t) - \log(f_{t-1})\right)$.}\\
\multicolumn{6}{l}{D refers transforming variable $f_t$  by computing $f_t - f_{t-1}$.}\\
\multicolumn{6}{l}{GG refers transforming variable $f_t$ by computing $0.1226\times 100\log(f_t/100)$ (as in \citet{Kleibergen:2009do, Gali:1999tx}).}\\

\end{tabular}
\end{table}

\restoregeometry